\documentclass[lettersize,journal]{IEEEtran}
\usepackage{amsmath,amsfonts}
\usepackage{amssymb}
\usepackage{amsthm}
\usepackage{mathtools}
\usepackage{ulem}
\usepackage{array}
\usepackage{tikz}
\usepackage{graphicx}
\usepackage{textcomp}
\usepackage[justification=justified]{caption}

\usepackage{subcaption}
\usepackage{stfloats}
\usepackage{color}
\usepackage{url}
\usepackage{verbatim}
\usepackage[linesnumbered,ruled,vlined]{algorithm2e}

\usepackage{titlesec}

\titlespacing*{\subsection}
  {0pt}{10pt}{5pt}

\newif\ifArxiv
\Arxivtrue

\usepackage{balance}

\newcommand{\new}[1]{{\color{black} #1}}
\newcommand{\Iff}{\mathrel{\text{iff}}}

\newtheorem{definition}{Definition}

\newtheorem{remark}{Remark}
\newtheorem{theorem}{Theorem}
\newtheorem{corollary}{Corollary}

\RestyleAlgo{ruled}

\begin{document}

\title{Avoiding Deadlocks Is Not Enough:\\ Analysis and Resolution of Blocked Airplanes}

\author{Shuhao Qi$^{1}$, Zengjie Zhang$^{1}$, Zhiyong Sun$^{2}$ and Sofie Haesaert$^{1}$ \vspace{-1em}
\thanks{This work was supported by the European project SymAware under the grant No. 101070802, the European project COVER under the grant No. 101086228, and the Dutch NWO Veni project CODEC
under grant No. 18244.}
\thanks{$^{1}$ S. Qi, Z. Zhang, S. Haesaert are with the Department of Electrical Engineering, Eindhoven University of Technology, Eindhoven, The Netherlands.
        {\tt\small \{s.qi, z.zhang3, s.haesaert\}@tue.nl}}
\thanks{$^{2}$Z. Sun is with the College of Engineering, Peking University, Beijing, China. {\tt\small \{zhiyong.sun@pku.edu.cn\}}}
}

\maketitle

\begin{abstract}
This paper is devoted to the analysis and resolution of a pathological phenomenon in airplane encounters called \textit{blocking mode}. As autonomy in airplane systems increases, a pathological phenomenon can be observed in two-aircraft encounter scenarios, where airplanes stick together and fly in parallel for an extended period. This parallel flight results in a temporary blocking that significantly delays progress. In contrast to widely studied \textit{deadlocks} in multi-robot systems, such transient blocking is often overlooked in existing literature. Since such prolonged parallel flying places high-speed airplanes at elevated risks of near-miss collisions, encounter conflicts must be resolved as quickly as possible in the context of aviation. We develop a mathematical model for a two-airplane encounter system that replicates this blocking phenomenon. Using this model, we analyze the conditions under which blocking occurs, quantify the duration of the blocking period, and demonstrate that the blocking condition is significantly less restrictive than that of deadlock. Based on these analytical insights, we propose an intention-aware strategy with an adaptive priority mechanism that enables efficient resolution of ongoing blocking phenomena while also incidentally eliminating deadlocks. Notably, the developed strategy does not rely on central coordination and communications that can be unreliable in harsh situations. The analytical findings and the proposed resolution strategy are validated through extensive simulations.
\end{abstract}


\section{Introduction}

Increasing the autonomy of airplane systems is expected to enhance airspace safety and efficiency~\cite{yang2021autonomous, de2021decentralized}. On-board detect-and-avoid (DAA) systems~\cite{de2023analyzing, nlr24report} have been developed to autonomously detect potential midair collisions and provide avoidance instructions based on local sensor data. As demand for aviation services grows, the increasingly crowded airspace is leading to more frequent airplane encounters~\cite{cohen2021urban, kochenderfer2008comprehensive}. To rigorously assess whether a DAA system can safely handle two-airplane encounters, extensive testing in high-fidelity Monte Carlo simulators~\cite{kochenderfer2008comprehensive,jenie2017safety} is necessary. Recently, the Netherlands Aerospace Centre reported a pathological phenomenon in DAA system tests~\cite{nlr24report} during a two-airplane encounter, where both airplanes repeatedly strive but fail to bypass each other, resulting in extended parallel flying, as shown in Fig.~\ref{fig:blocking_sys}. This pathological phenomenon causes airplanes to deviate from their intended paths and significantly delays target attainment. Moreover, the report~\cite{nlr24report} indicates that under realistic uncertainty, such parallel high-speed flying places airplanes at high risk of near-miss collisions when conflicts are not resolved promptly. Therefore, it is crucial to analyze the underlying reason for inducing this phenomenon in two-airplane encounters and develop a resolution strategy.

\begin{figure}[tb]
    \centering
    \includegraphics[width=0.45\textwidth]{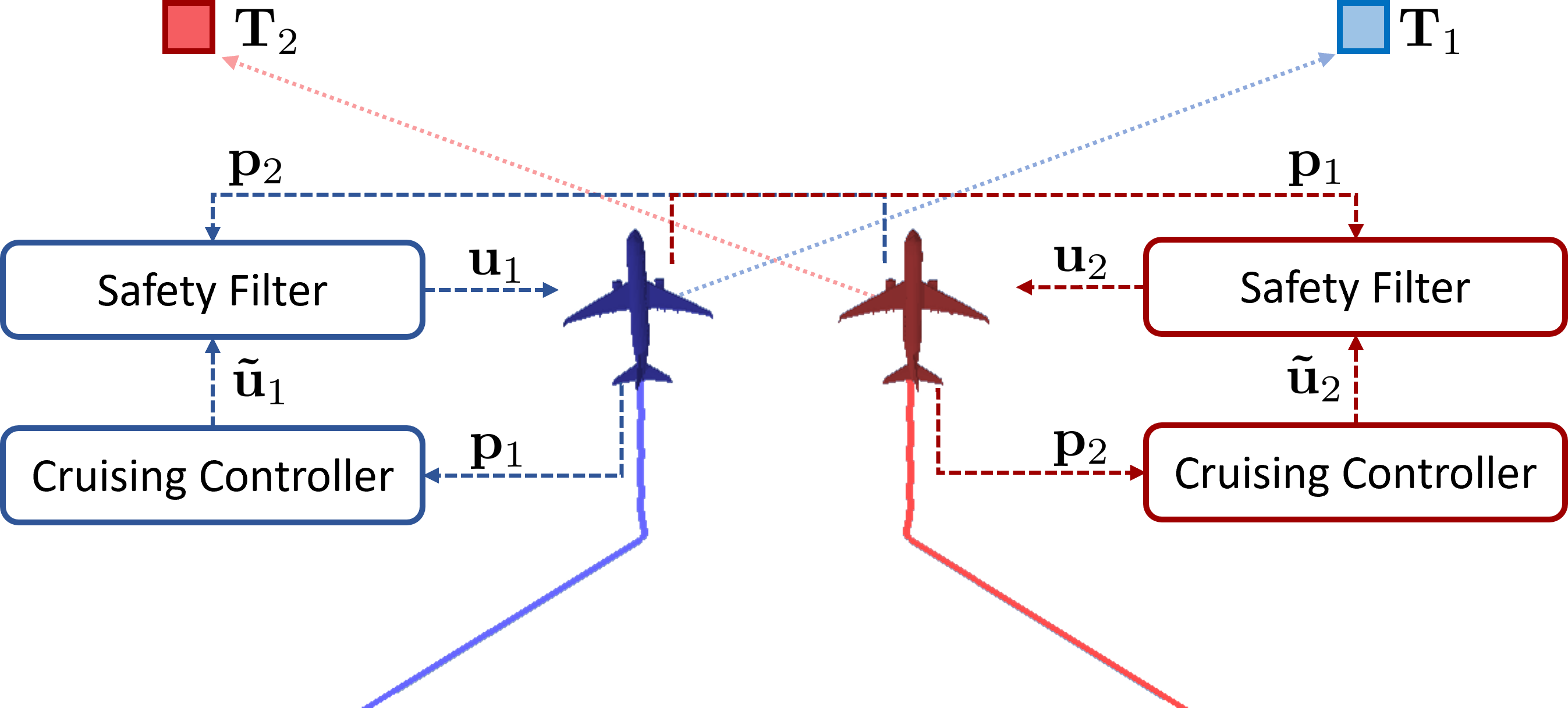}
    \caption{Blocking phenomenon in a two-airplane system. Two airplanes fly toward target positions indicated by square markers. Their trajectories, shown as solid lines, converge to parallel paths as they approach each other.}
    \label{fig:blocking_sys}
\end{figure}

Such pathological parallel flight behavior arises from prolonged, unresolved encounter conflicts. In these situations, the DAA systems of both airplanes are activated to maintain a safe distance, but this compromises progress toward their respective targets. Two typical behaviors in which task completion is compromised are \textit{deadlock} and \textit{livelock}. Deadlock~\cite{jankovic2023multiagent, grover2023before, chen2024deadlock} refers to situations in which agents get stuck in a static configuration and never reach targets, whereas livelock~\cite{abate2009understanding} describes scenarios in which agents keep moving but never reach targets. In contrast to livelock, which is rarely observed, deadlock has been extensively studied in the literature on multi-robot systems. 
Since safety is typically prioritized over task completion, theoretical results have shown that multi-robot systems with decentralized safety-critical controllers, such as control barrier functions (CBFs)~\cite{wang2017safety, grover2023before}, model predictive control (MPC)~\cite{chen2023multi}, and safe-reachable set~\cite{ouyang23cdc} tend to cause deadlocks. From a control-theoretic perspective, deadlock can be viewed as an undesired equilibrium~\cite{jankovic2023multiagent, reis2020control} in which system stability is compromised to prioritize safety. Although deadlocks have been widely studied, most discussions focus on mobile robots, in which each robot can stop to ensure safety. Moreover, deadlocks in such systems occur only under restrictive conditions, such as highly symmetric configurations~\cite{grover2023before, grover2020does}.

The concepts of deadlock and livelock were originally studied in queueing systems within software engineering~\cite{tanenbaum2009modern}. Another relevant concept from queueing systems is \textit{blocking}~\cite{mannucci2021provably, yu2021distributed}, which refers to situations in which the progress of agents is temporarily halted for a finite period. Inspired by this, we refer to the finite-time parallel flying behavior observed in Fig.~\ref{fig:blocking_sys} as \textit{blocking mode} in continuous-state systems. Unlike endless deadlock and livelock, the blocking behavior illustrated in Fig.~\ref{fig:blocking_sys} causes only a finite-time delay in progress and has therefore been largely overlooked in the literature on mobile robots. Although resolving a conflict inevitably delays the agents' progress, the prolonged parallel flight is unnecessary and can be eliminated with a faster resolution. In addition, unlike mobile robots, airplanes must maintain high speeds and cannot stall to ensure safety, making such parallel-flying behavior both dangerous and unacceptable for airplanes. Moreover, the conditions under which it occurs are far less restrictive than those required for deadlock in mobile robots, a claim that will be validated in this paper. Additionally, we tested several common avoidance controllers in airplane encounter scenarios, including the CBF‐based safety filter~\cite{CSM23}, velocity obstacle~\cite{van2008reciprocal}, and potential field~\cite{potential1991}, and observed similar blocking phenomena across all of them. These reasons indicate that blocking is a fundamental issue for airplanes. In this paper, we formally analyze the blocking phenomenon using the CBF‑based safety filter, which naturally aligns with the design principles of DAA systems and offers an analytically tractable formulation with safety guarantees. Analysis of this formulation reveals that blocking typically arises when two airplanes simultaneously choose symmetric avoidance maneuvers under the optimality principle, preventing them from reaching a collaborative resolution. If aircraft can achieve a collaborative resolution, the conflicts can be resolved efficiently. The blocking phenomenon should be explicitly addressed in the design of airplane systems.

Although blocking resolution has not been well studied, existing deadlock-resolution approaches provide promising ways to address the blocking phenomenon. Deadlock resolution methods can be mainly grouped into four categories:  1) Central coordination, in which a central coordinator manages agent interactions to resolve deadlocks~\cite{mo16cdc}. However, such a central coordinator scales poorly with increasing traffic volumes~\cite{tomlin98TAC, pritchett2017negotiated}. 2) Reactive distributed controllers rely on local communication to negotiate with nearby airplanes, using techniques such as parametric control Lyapunov functions~\cite{weng2022convergence} and rotation controllers for position swaps~\cite{grover2023before, Arul2021iros}. However, local communication is vulnerable in practice to time delays, language ambiguity, and disturbances, particularly among different types of aircraft~\cite{master2020}. Given the safety-critical nature of aviation, airplanes must ensure both safety and task accomplishment, even in the absence of communication. 3) Perturbation addition introduces disturbance to break strict deadlock conditions~\cite{wang2017safety}. While perturbation addition is practical, it does not guarantee deadlock resolution and can compromise safety constraints, making it even less suitable for addressing the less restrictive blocking behavior. 4) Predefined rules and priorities, such as right-hand priority rules~\cite{pierson2020weighted}, can be encoded directly into controllers. However, these fixed approaches are vulnerable in uncertain environments and can result in unnecessary energy consumption.

Two recent airplane crashes in the U.S.~\cite{cnn2025jan, cnn2025feb} were caused by delayed communications and the absence of a central controller, highlighting their unreliability. Moreover, fixed rules and priorities are inefficient. This raises a fundamental question: In the absence of communication and central controllers, can airplanes still make adaptive decisions to resolve conflicts? Designing a fully decentralized resolution framework under these harsh conditions remains a significant challenge, even for two-airplane airplane encounters. However, human walkers and drivers always find a quick resolution way in collaboration with their neighbors, even in crowded situations and without communication, by exhibiting their intentions and estimating the intentions of surrounding agents~\cite{mavrogiannis2023core}. Inspired by this insight from social navigation, we propose an intention-aware strategy that generates an adaptive priority of unblocking behavior, which provides guarantees and also improves efficiency. Although the proposed method is protocol-based, it demonstrates the potential for efficient and decentralized resolution without communication, laying the groundwork for future research.

In this paper, vector variables will be represented as bold symbols (e.g., $\mathbf{x} \!\in\! \mathbb R^n$), while scalars will be written as $x\!\in\!\mathbb R$. $\|\cdot\|$ denotes the Euclidean norm of a vector. Furthermore, we define the following angular normalization operator that maps an angle $a \!\in\!\mathbb R$ to the range $[-\pi, \pi)$,
\begin{equation}
    \measuredangle(a) = (a + \pi) \% 2 \pi - \pi,
    \label{eq:operator}
\end{equation}
\noindent where $\%$ is the modulo operator.

\section{Problem Statement}
This paper considers a horizontal two-airplane encounter, a typical setting used to analyze aerial encounters~\cite{kochenderfer2008comprehensive}. \new{We consider a two-airplanes system, denoted as $\mathcal{A} \!=\! \{A_1, A_2\}$ that flies at the same forward speed. The horizontal behavior of an airplane $A_i\!\in\!\mathcal{A}$ is commonly characterized by a unicycle model \cite{sun2021collaborative, zhou2021distributed}}. Given position $\mathbf{p}_i(t)\!=\![p_i^x(t), p_i^y(t)]^T\!\in\!\mathbb{R}^2$ and heading angle $\theta_i(t)$, we define a unicycle model for $A_i$ with a constant forward speed $v\!>\!0$ as 

 	\begin{equation}
 		\begin{bmatrix} \dot{p}_i^x \\ \dot{p}_i^y \\ \dot{\theta}_i \end{bmatrix} = \begin{bmatrix}
 			v \cos (\theta_i)\\
 			v \cos (\theta_i)\\
 			a_i
 		\end{bmatrix},\quad i\in\{1,2\}.
 		\label{eq:uni_sys}
 	\end{equation}
    
Further, we assume that each airplane $A_i$ can accurately observe both its own and the other airplanes' positions and heading angles. Since airline routes consist of a sequence of waypoints connected via airways, the overall route-tracking task can be decomposed into a series of target-reaching problems~\cite{ren2017air}. Therefore, in an encounter scenario, we assume each airplane is tasked with reaching its respective destination point $\mathbf{T}_i$ for $i\!\in\!\{1, 2\}$. Fig.~\ref{fig:cruise} depicts the relevant variables of a two-airplane system. Two airplanes may encounter each other if an intersection point exists between their cruising paths, such as $\mathbf{p}_c$ in Fig.~\ref{fig:cruise}. The distance between two airplanes must satisfy $\|\mathbf{p}_1(t) - \mathbf{p}_2(t)\| \! \geq \! r$ for all $t\!\geq\!0$, where $r$ denotes the safe margin. We assume the initial and target positions of two airplanes satisfy the safety constraints, i.e.,  $\| \mathbf{p}_1(0) \!-\! \mathbf{p}_2(0) \| \! \!\geq\!  r $ and $\| \mathbf{T}_1 \!-\! \mathbf{T}_2 \| \!\geq\!r$.

\begin{figure}[htp]
	\centering
	\includegraphics[width=0.3\textwidth]{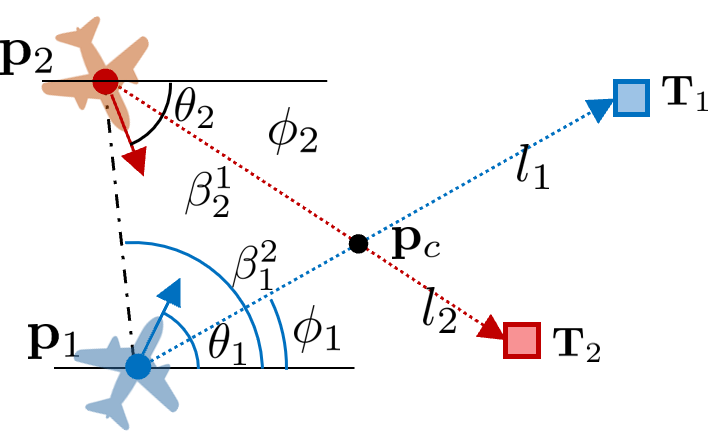}
	\caption{Encounter scenario of a two-airplane system. The solid dots represent the current positions of airplanes, $\mathbf{p}_1$ and $\mathbf{p}_2$, while rectangle markers indicate their target positions $\mathbf{T}_1$ and $\mathbf{T}_2$. For each airplane, three key angles are shown: the heading angle $\theta_{i}$, the bearing angle $\beta_i^j$, and the cruising angle $\phi_i$. Line segments $l_1$ and $l_2$ connect airplanes and their respective targets. The intersection point between $l_1$ and $l_2$ is $\mathbf{p}_c$. } 
	\label{fig:cruise}
\end{figure}

Each airplane is equipped with a cruising controller and a DAA system. The cruising controller is responsible for navigating the airplane $A_i$ toward its target $\mathbf{T}_i$. The DAA system typically monitors for unsafe cruising control inputs and provides safe operational guidance, including heading angles and vertical rates~\cite{nlr24report}, to maintain a safe separation from other airplanes. However, since this paper focuses on horizontal behavior for simplicity, the DAA system in this study only corrects the heading angle. As reported by~\cite{nlr24report}, such a framework can inadvertently lead to prolonged parallel flight, when the cruising heading angles of both airplanes are adjusted by their DAA systems, yet the adjusted safe headings fail to resolve the encounter conflict promptly. To better understand this pathological behavior, this paper aims to formally characterize the parallel-flying phenomenon and analyze its underlying causes. In a two-airplane encounter, conflicts should ideally be resolved quickly if both airplanes can adopt collaborative bypass maneuvers. Since communications and central coordination are unreliable in harsh situations, our goal is to develop a resolution strategy that is both efficient and provably safe, even in the absence of either central coordination or inter-airplane communication.

\section{Modelling blocking}\label{sec:model_blocking}
\new{In this work, for simplicity, we will focus on designing a desired heading angle $\theta^*$ to resolve encounter conflicts. We assume that a high-gain controller $a_i=-k(\theta_i-\theta_i^{*})+\dot \theta_i^{*}$ is implemented to track desired heading angle $\theta^*$, where $k$ is a positive gain large enough such that $\theta_i$ converges to $\theta_i^{*}$ almost instantaneously (i.e., $\theta_i\approx \theta_i^{*}$). Under this setting, model~\eqref{eq:uni_sys} can be simplified to a single integrator with constant speed,
\begin{equation}
    \quad  \mathbf{\dot{p}}_i  =  \textstyle \underbrace{\begin{bmatrix}
         \, v\text{cos}(\theta_i)\\
         \, v\text{sin}(\theta_i)\\
    \end{bmatrix}}_{\mathbf{u}_i}, \ i \in \{1,2\},
    \label{eq:sys}
\end{equation} where $\mathbf{u}_i\!\in\!\mathbb{R}^2$ denotes the horizontal velocity of $A_i$. The following discussion is based on heading angle control for the simplified model in Eq.~\eqref{eq:sys}. This simplification aligns with the operations of DAA systems, which typically provide heading angle adjustments to ensure safety. In this section, we present an analyzable control framework that emulates the functionality of actual airplane systems and replicates the blocking phenomenon, and subsequently provide a formal definition for the blocking phenomenon.}

\subsection{Cruising controller}
The cruising controller for the airplane model in Eq.~\eqref{eq:sys} is defined as follows,
\begin{equation}
    \mathcal{N}(\mathbf{T}_i, \mathbf{p}_i): \tilde{\mathbf{u}}_i = \textstyle      \textstyle v \frac{\mathbf{T}_i - \mathbf{p}_i}{\|\mathbf{T}_i - \mathbf{p}_i\|} = v \begin{bsmallmatrix} \text{cos}(\phi_i)\\ \text{sin}(\phi_i)
\end{bsmallmatrix}, 
\label{eq:u_n}
\end{equation} 
for $i \in \{1, 2\}$ and $ \mathbf{p}_i \neq \mathbf{T}_i$. The angle $\phi_i$ denotes the cruising angle heading toward the target point, as illustrated in Fig.~\ref{fig:cruise}. We say $A_i$ is in \textit{cruising mode} when $\mathbf{u}_i\!=\!\tilde{\mathbf{u}}_i$, i.e., its heading angle equals the cruising angle $\theta_i \!=\! \phi_i$.

\subsection{CBF-based safety filter}

Realistic DAA systems typically rely on complex methodologies, such as dynamic programming and extensive lookup tables~\cite{daa-eurocae, nlr24report}, which pose significant challenges for formal analysis. Therefore, we instead employ an analytically tractable safety filter~\cite{CSM23} for the model in~\eqref{eq:sys}, which emulates the functionality of DAA systems to ensure safety.

\new{From a control-theoretic perspective, the safety property is formally captured by forward invariance~\cite{wang2017safety}. CBFs~\cite{ames2019control} are a convenient approach for enforcing forward invariance in a dynamical system and are defined as follows,
\begin{definition}[Control barrier function, CBF]
Let $\mathcal{C} \subset \mathbb{R}^n$ be a super level set of a continuously differentiable function $h: \mathbb{R}^n \!\rightarrow\!\mathbb{R}$. Then $h$ is a valid CBF that ensures the forward invariance of $\mathcal{C}$ for a nonlinear system $\dot{x}\!=\!f(x,u)$ if there exists a continuous function $\alpha(\cdot):\mathbb{R} \!\rightarrow \!\mathbb{R}$ such that for all $x \!\in\!\mathbb{R}^n$, there exists a
control input $u\!\in\!\mathbb{R}^m$ satisfying: 
$$
\textstyle \alpha(h(x)) + \frac{d }{d x}[h(x)] f(x,u)\geq 0. 
$$ 
where $\alpha(\cdot)$ is strictly monotonically increasing with $\alpha(0)\!=\!0$. 
\end{definition} 
\noindent In this paper, we define the CBF as $h(\mathbf{p}_1, \mathbf{p}_2)\!:=\!\|\mathbf{p}_1-\mathbf{p}_2\|^2 \!-\! r^2$ for the joint two-airplane system and, for simplicity, we use a linear function $\alpha(x)\! = \! \alpha x, \alpha\!\in\!\mathbb R^+$. The CBF conditions for the collision-free set $\textstyle \{ (\mathbf{p}_1, \mathbf{p}_2) \mid h(\mathbf{p}_1, \mathbf{p}_2)\!\geq\! 0\}$ is derived as follows,
\begin{equation}
    \textstyle \alpha h(\mathbf{p}_1, \mathbf{p}_2) + 2(\mathbf{p}_1 - \mathbf{p}_2)^T (\mathbf{u}_1 - \mathbf{u}_2)  \geq 0.
    \label{eq:central_cbf}
\end{equation}
\noindent The inputs $\mathbf{u}_1$ and $\mathbf{u}_2$ satisfying the above condition can ensure the safety. Note that the above condition relies on control inputs of the other airplane, implying a centralized control manner. Referring to~\cite{grover2023before,wang2017safety}, we derive a decentralized CBF condition based on a half-responsibility separation,
\begin{equation}
    g(\mathbf{p}_i, \mathbf{p}_j, \mathbf{u}_i) := \textstyle \frac{\alpha}{2} h(\mathbf{p}_1, \mathbf{p}_2)\! +\! 2(\mathbf{p}_i\!-\!\mathbf{p}_j)^T \! \mathbf{u}_i \geq 0,
    \label{eq:cbf}
\end{equation}
for $i \!\neq\! j$  and  $i,j \!\in\! \{1,2\}$.} The addition of decentralized CBF conditions of two airplanes is equivalent to the centralized condition in Eq.~\eqref{eq:central_cbf}. Using the CBF condition as a safety constraint, a safety filter is formulated as a quadratic programming problem,
\begin{subequations}
\begin{align}
 \mathcal{F}(\mathbf{p}_i, \mathbf{p}_j, \mathbf{\tilde{u}}_i)\!: \ &  \textstyle \underset{\mathbf{u}_i \in \mathbb{R}^2}{\operatorname{argmin}} && \hspace{-5mm} \frac{1}{2}\|\mathbf{u}_i-\mathbf{\tilde{u}}_i \|^2  \label{sf:obj} \\
& \text{ s.t.} && \hspace{-5mm}  \textstyle g(\mathbf{p}_i, \mathbf{p}_j, \mathbf{u}_i) \geq 0 \label{sf:cbf} \\ 
& && \hspace{-5mm} \|\mathbf{u}_i\| = v, 
\label{sf:norm} 
\end{align}
\label{eq:sf}
\end{subequations} \vspace{-2mm}

\noindent for $i \neq j$  and  $i,j \in \{1,2\}$. 
Similar to the principles in DAA systems, the safety filter minimally adjusts them to ensure safety if cruising control inputs are discerned as unsafe; otherwise, it maintains the original cruising control law. In this paper, the safety filter is regarded as being \textit{activated} when $\mathcal{F}(\mathbf{p}_i, \mathbf{p}_j, \mathbf{\tilde{u}}_i) \!\neq\! \mathbf{\tilde{u}}_i$. 

\begin{figure}[t]
  \centering
  \begin{subfigure}[b]{0.5\textwidth}
    \centering
    \includegraphics[width=0.75\textwidth]{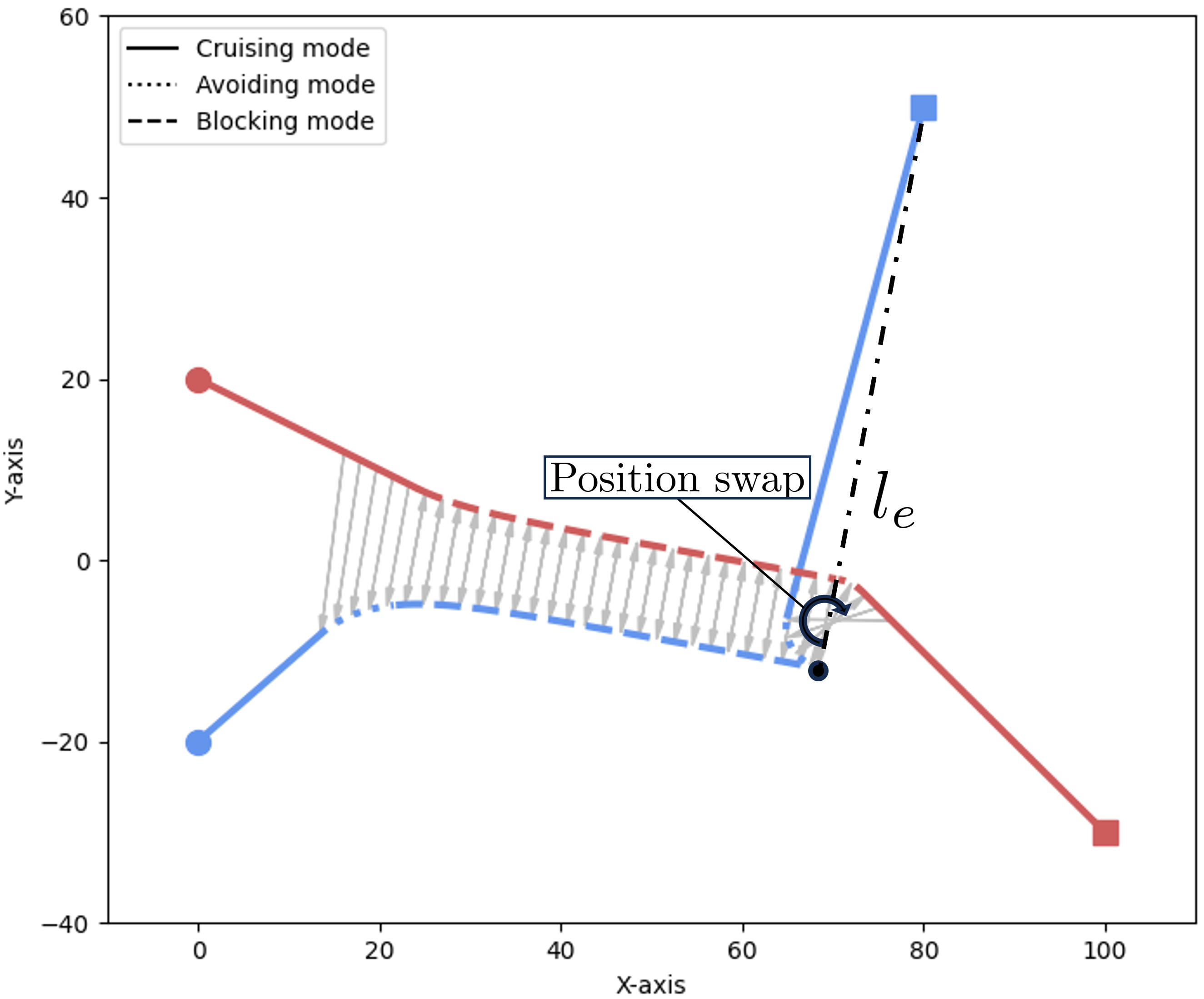}
    \caption{Trajectories of two airplanes in blocking mode. Gray arrows represent the activation of safety filters.}\vspace{2mm}
    \label{fig:blocking_traj}
  \end{subfigure} 

  \begin{subfigure}[b]{0.5\textwidth}
    \centering
    \includegraphics[width=0.9\textwidth]{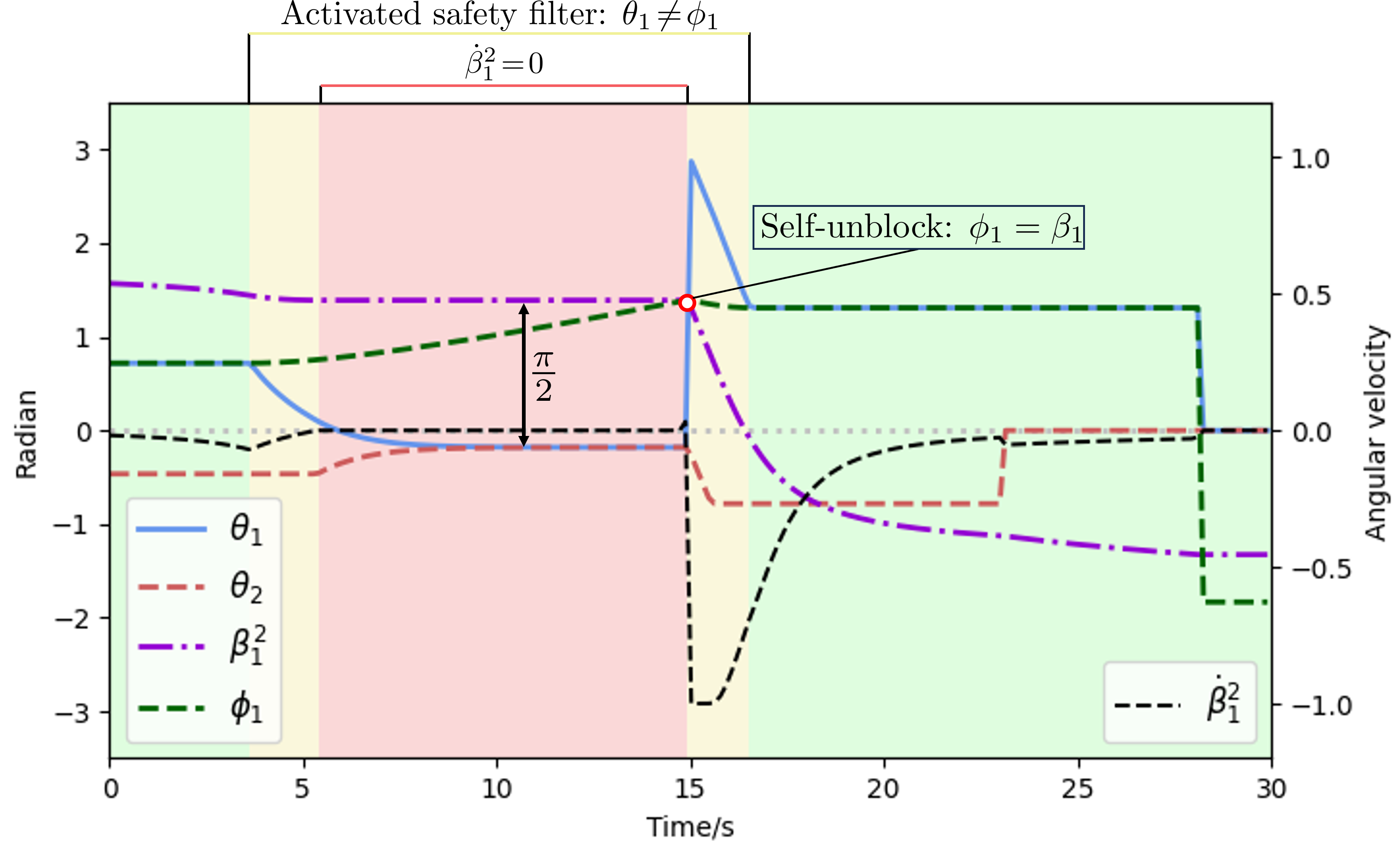}
    \caption{Curves of related variables. Light green, yellow, and red areas denote the cruising, avoiding, and blocking modes of $A_1$, respectively.}
    \label{fig:blocking_curve}
  \end{subfigure}
  \caption{The process of an encounter with a blocking mode.}
  \label{fig:blocking_sim}
\end{figure}

\subsection{Blocking mode}

By assuming the safety filters of two airplanes have uniform parameters of $\alpha$ and $r$, the framework consisting of cruising controllers in~\eqref{eq:u_n} and safety filters in~\eqref{eq:sf} replicates a phenomenon shown in Fig.~\ref{fig:blocking_traj}, similar to the parallel flying observed in~\cite{nlr24report}. In the following, we formally characterize this phenomenon on the basis of the above formulation.
First, we define deadlock and livelock as follows.
\begin{definition}[Deadlock~\cite{grover2023before} and Livelock~\cite{abate2009understanding}]
    An agent with index $i$ is in a deadlock at time $t_0$ if the agent remains stationary at a point different from its target, i.e., $\forall \, t > t_0$, $\mathbf{p}_i(t) = \mathbf{p}_i(t_0) \neq \mathbf{T}_i$. An agent is in a livelock at time $t_0$ if 
    the agent continues to move indefinitely without reaching its target, i.e., $\forall \, t > t_0$, $\mathbf{p}_i(t) \neq \mathbf{T}_i$ and $\dot{\mathbf{p}}_i(t) \neq 0$. 
\end{definition}

\noindent  \new{Let $\beta_i^j$ denote the bearing angle from $A_i$ to $A_j$ as illustrated in Fig.~\ref {fig:cruise}, such that $\textstyle \begin{bsmallmatrix} \text{cos}(\beta_i^j)\\ \text{sin}(\beta_i^j)
\end{bsmallmatrix}\!=\!\frac{\mathbf{p}_j - \mathbf{p}_i}{\|\mathbf{p}_j - \mathbf{p}_i\|}$.
We define the parallel flying phenomenon in Fig.~\ref{fig:blocking_sys} as \textit{blocking mode},}

\new{
\begin{definition}[Blocking mode]\label{def:blocking}
Considering a two-airplane system $\mathcal{A}\!=\!\{A_1, A_2\}$, $A_i$ is in blocking mode at $t_0$ if it deviates from its intended cruising direction to ensure safety by activating the safety filter, i.e.,
$\mathcal{F}(\mathbf{p}_i, \mathbf{p}_j, \mathbf{\tilde{u}}_i)\!\neq\! \tilde{\mathbf{u}}_i$, and maintains a constant bearing relative to another airplane, i.e.,
$\dot{\beta}_i^j\!=\!0$, where $i,j\!\in\!\{1,2\}$ and $i\!\neq\!j$.  
\end{definition}}
\noindent \new{The derivative of $\beta_i^j$ can be calculated by $\dot{\beta}_i^j\!:=\!\frac{ (\mathbf{p}_j - \mathbf{p}_i) \times (\mathbf{u}_j - \mathbf{u}_i)}{\|\mathbf{p}_j - \mathbf{p}_i\|^2}$, which is a scalar, due to the cross product in two-dimensional space.} Contrary to the blocking mode with $\dot{\beta}_i^j \!=\! 0$, $A_i$ is regarded as being in \textit{avoiding mode} if $\mathcal{F}(\mathbf{p}_i, \mathbf{p}_j, \mathbf{\tilde{u}}_i)\!\neq\! \tilde{\mathbf{u}}_i(t)$ and $\dot{\beta}_i^j(t) \neq 0$, this indicates that the airplane performs a rotation around the other airplane. 
Such rotation around each other to swap positions is widely used in conflict resolution~\cite{grover2023before, weng2022convergence, Arul2021iros}.  Therefore, the relative angular velocity $\dot{\beta}_i^j$ is indicative for an effective resolution of the encounter conflict. As shown in Fig.~\ref{fig:blocking_curve} and Fig.~\ref{fig:switch}, an airplane will switch among the three modes in an encounter.

\new{
\begin{remark}
Since most safety controllers function similarly by preventing agents from approaching each other, blocking is not unique to the CBF-based safety filter. It also occurs in commonly used controllers such as velocity obstacle~\cite{van2008reciprocal} and potential fields~\cite{potential1991}, as detailed in \ifArxiv Appendix~\ref{ap:vo}.
\else Appendix A of~\cite{arxiv_version}.\fi 
\end{remark} }

\begin{figure}[t]
    \centering
\includegraphics[width=0.48\textwidth]{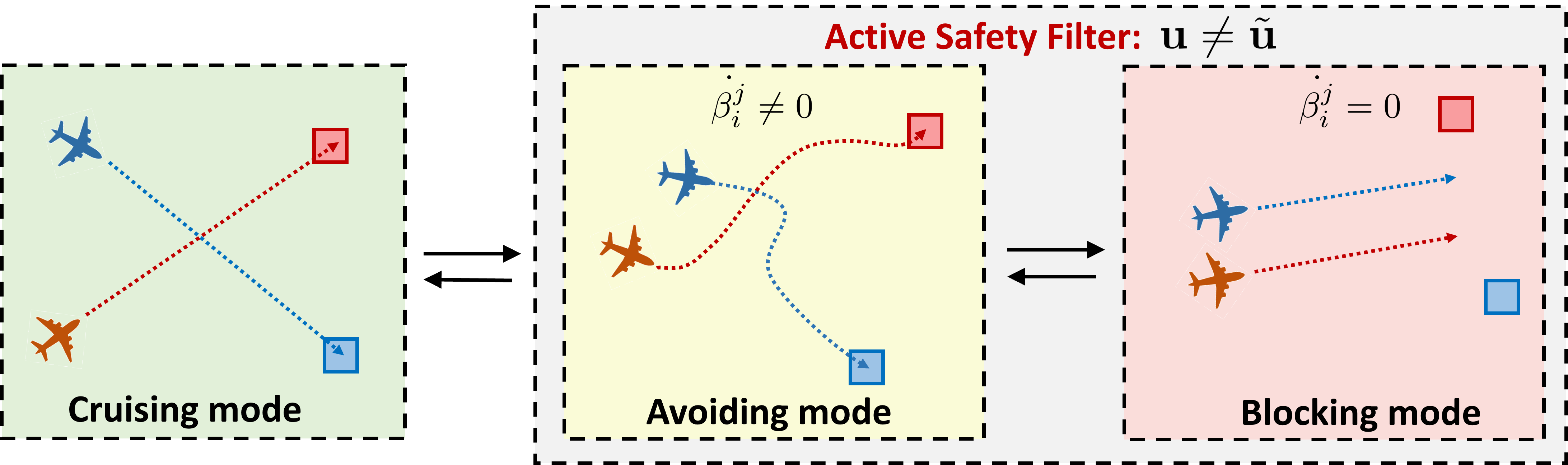}
    \caption{Three modes of an airplane in an encounter.}
    \label{fig:switch}
\end{figure}




\section{Analysis of Blocking Mode}~\label{sec:condition}
This section begins by presenting an explicit solution of the safety filter in Eq.~\eqref{eq:sf}, illustrating how safety filters ensure inter-airplane safety during flight. Using the explicit solution, we derive the conditions under which blocking modes occur. This allows us to estimate the blocking duration to quantify its potential impact in advance. 
\subsection{Explicit solutions of safety filter}
When two airplanes are sufficiently far apart, airplanes can select any admissible control inputs, meaning all the inputs in the set $\{\mathbf{u} \!\mid\!\|\mathbf{u}\|\!=\!v\}$ satisfy the CBF condition in Eq.~\eqref{sf:cbf}. This situation is referred to as the \textit{free flight} phase~\cite{tomlin98TAC}. \new{Based on Eq.~\eqref{eq:cbf}, the distance condition for the free flight phase is }
$$
 \|\mathbf{p}_2 - \mathbf{p}_1\| > \frac{2v}{\alpha} + \sqrt{\frac{4v^2}{\alpha^2} + r^2}.
$$

\noindent When the airplanes come closer, not all heading angles remain safe. Fig.~\ref{fig:solution} illustrates how the safety filter functions. The green arc denotes the set of safe heading angles, while the red and yellow arcs represent the set of unsafe heading angles. The red and yellow arcs indicate the negative and the positive directions, respectively. They are symmetrically positioned relative to the line between two airplanes, $l_a$, each encompassing an identical angular range denoted as $\Delta$. Hence, the unsafe heading set is characterized by $\{\theta_i \mid\|\measuredangle(\theta_i-\beta_i^j)\|\!<\!\Delta\}$. In the free flight phase, $\Delta\!=\!0$. As the airplanes approach each other, $\Delta$ increases and gradually converges to $\frac{\pi}{2}$, as illustrated in Fig.~\ref{fig:blocking_curve}. Due to the minimal interference manner, the safety filter will correct to the nearest boundary point in the safe set, i.e., $\beta^j_i\!+\!\Delta$ or $\beta_i^j\!- \!\Delta$ if the cruising angle $\phi_i$ is unsafe. The following explicit solution fully captures the function of safety filters,
\begin{theorem}
Considering a two-airplane system $\mathcal{A}\! =\! \{A_1, A_2\}$ with dynamics in Eq.~\eqref{eq:sys} and cruising controllers in Eq.~\eqref{eq:u_n}, the explicit solution to the safety filter of $A_i$ in Eq.~\eqref{eq:sf} is as follows,
\begin{equation}
\begin{aligned}
      \mathbf{u}_i & = v[\text{cos}(\theta_i), \text{sin}(\theta_i)]^T \\
      {\theta}_i &= 
    \begin{cases}
      \beta_i^j \pm \Delta, \, & \phi_i = \beta_i^j \\
    \beta_i^j - \Delta, & \measuredangle(\phi_i - \beta_i^j) \in (-\Delta, 0)  \\
    \beta_i^j + \Delta, & \measuredangle(\phi_i - \beta_i^j) \in (0, \Delta)  \\
    \phi_i, & otherwise\\
    \end{cases}
\end{aligned}
\label{eq:sol}
\end{equation} 
where 
$\Delta\!=\! \text{cos}^{-1} \left(\text{min}(1, \frac{\alpha h(\mathbf{p}_1, \mathbf{p}_2)}{4v \|\mathbf{p}_1 -\mathbf{p}_2\|}) \right)\! \in\! [0, \frac{\pi}{2}]$, and $i,j \!=\! \{1,2\}$ with $i\! \neq\! j$. 
\label{th:sol}
\end{theorem}
\begin{proof}\renewcommand{\qedsymbol}{}
The proof are provided in \ifArxiv Appendix~\ref{ap:thm_1}.
\else Appendix B of~\cite{arxiv_version}.\fi
\end{proof}

\begin{figure}[tb]
    \centering
    \includegraphics[width=0.34\textwidth]{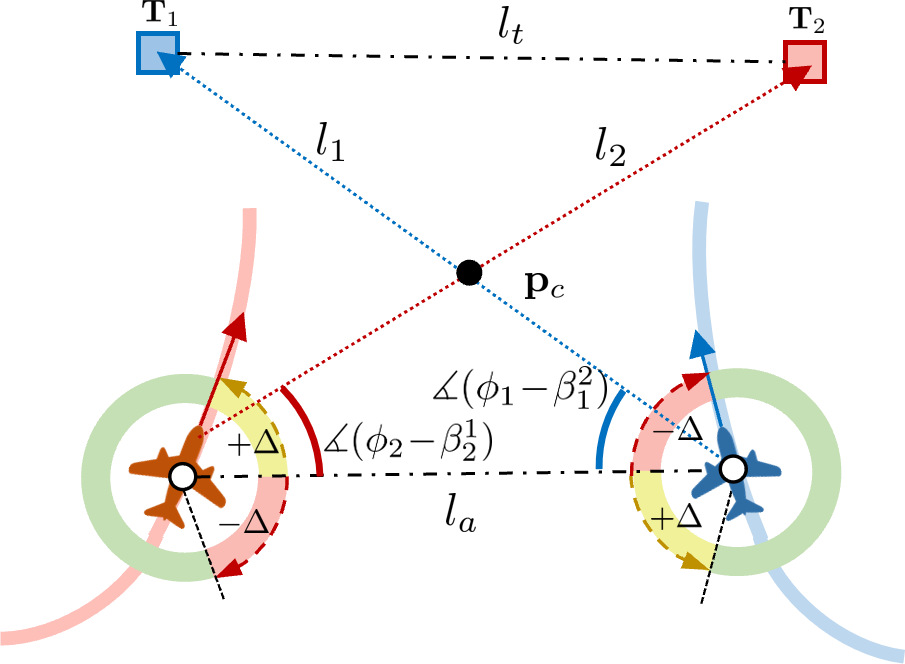}
    \caption{Sketch of explicit solution for the safety filter. \new{For each $A_i$, the green arc denotes the set of safe heading angles $\theta_i$, while the red and yellow arcs represent the set of unsafe heading angles. } The line segment $l_a$ connects $\mathbf{p}_1$ and $\mathbf{p}_2$, while $l_t$ connects $\mathbf{T}_1$ and $\mathbf{T}_2$. As for $i=1,2$, $l_i$ connects $\mathbf{p}_i$ and $\mathbf{T}_i$. $\mathbf{p}_c$ is the intersection point of $l_1$ and $l_2$.}
    \label{fig:solution}
\end{figure}

\subsection{Blocking condition}

According to Def.~\ref{def:blocking}, two necessary characteristics of the blocking mode include: a zero derivative of the bearing angle and an active safety filter. Given that the safety filter ensures the inter-airplane distance $||\mathbf{p}_1 - \mathbf{p}_2|| \!\geq\! r$ remains safe during flight, the derivative of the bearing angle, defined as $\dot{\beta}_i^j \!:=\! \frac{ (\mathbf{p}_j - \mathbf{p}_i) \times (\mathbf{u}_j - \mathbf{u}_i)}{\|\mathbf{p}_j - \mathbf{p}_i\|^2}$, equals zero only when the relative velocity $(\mathbf{u}_j - \mathbf{u}_i)$ is parallel to the relative position vector $(\mathbf{p}_1 - \mathbf{p}_2)$ or when the velocities are identical, i.e., $\mathbf{u}_j = \mathbf{u}_i$. Building upon the activation conditions of a safety filter presented in Theorem~\ref{th:sol}, we can now formalize the situations under which blocking mode occurs.
\begin{theorem} 
Consider a two-airplane system $\mathcal{A} = \{A_1, A_2\}$ where each airplane is equipped with a safety filter~\eqref{eq:sf}. For $i,j \!\in\!\{1, 2\}$ and $i\!\neq\!j$, 
\begin{itemize}
    \item[1)] Both $A_i$ and $A_j$ are in blocking mode (cf., Def.~\ref{def:blocking}) $\Iff$ $\exists s \!\in\! \{-1, 1\}$, $s\measuredangle(\phi_i-\beta_i^j) \in [0, \Delta)$ and $-s\measuredangle(\phi_j-\beta_j^i) \in  [0, \Delta)$;
    \item[2)] $A_i$ is in blocking mode and $A_j$ is in cruising mode $\Iff$ $\exists s\!\in\! \{-1, 1\}$, $s \measuredangle(\phi_i\!-\!\beta_i^j)\! \in\! [0, \Delta)$ and $\phi_j\! = \!\measuredangle(\beta_i^j \!+\! s \Delta)$.
\end{itemize}
\label{th:block}
\end{theorem} 
\begin{proof} \renewcommand{\qedsymbol}{}
\ifArxiv The proof can be found in Appendix~\ref{ap:thm_2}.
\else The proof is provided in Appendix C of~\cite{arxiv_version}.\fi 
\end{proof}

This theorem reveals that blocking occurs when the two airplanes' cruising angles are approximately mirrored. As illustrated in Fig.~\ref{fig:solution}, blocking arises when one airplane’s cruising angle falls within the yellow arc and the other’s lies within the red arc. In these situations, the two airplanes tend to choose mirror-image avoidance directions, because the objective function~\eqref{sf:obj} penalizes large heading deviations. Such individual optimality can lead to a greater overall loss when the two airplanes do not cooperate. If target locations cause the airplanes to approach each other symmetrically, blocking mode becomes inevitable. Consequently, a sufficient condition for blocking is formally stated in Corollary~\ref{th:block_bf}. This corollary can facilitate the advanced detection of an upcoming blocking mode before entering it.
\begin{corollary}
If both airplanes in $\mathcal{A} \!=\! \{A_1, A_2\}$ are in cruising modes, they will enter into the blocking modes defined in Def.~\ref{def:blocking} in the future if the target points $\mathbf{T}_1$ and $\mathbf{T}_2$ make $\measuredangle(\phi_1-\beta_1^2)\! =\! -\measuredangle(\phi_2-\beta_2^1)$ satisfied
and the encounter point $\mathbf{p}_c$ exists and lies equidistant from both airplanes, i.e., $\|\mathbf{p}_2 \! - \! \mathbf{p}_c\|\! =\! \|\mathbf{p}_1 \! - \! \mathbf{p}_c\|$.
\label{th:block_bf} 
\end{corollary}
\begin{proof} \vspace{-4mm}
The proof can be found in Appendix~\ref{ap:cor_1}.
\end{proof}

\subsection{Estimation of blocking duration}

Having discussed how the safety filter causes blocking modes, we aim to estimate their duration. To achieve this, we need to analyze airplane behavior during blocking modes and identify the conditions under which they end.

As shown in Theorem~\ref{th:block}, the conditions for blocking modes depend on the $\phi_i$, $\beta_i^j$, and $\Delta$ for $i\!\in\!\{1,2\}$. Fig.~\ref{fig:blocking_curve} shows the curves of $\phi_1$, $\theta_1$, $\beta_1^2$, $\dot{\beta}_1^2$, and $\theta_2$ during an encounter scenario in which blocking occurs. Because the solution $\mathbf{p}_i$ to the system in Eq.~\eqref{eq:sys} is continuous over time, $\phi_i$, $\beta_i^j$, and $\Delta$, as continuous functions of $\mathbf{p}_i$ and $\mathbf{p}_j$, evolve continuously as well. When an airplane enters a blocking mode, $\dot{\beta}_i^j\!=\!0$ implies that $\beta_i^j$ remains fixed, and $\Delta$ remains fixed or converges to $\frac{\pi}{2}$, as shown in Fig.~\ref{fig:blocking_curve}. Consequently, breaking the blocking condition depends solely on the change of $\phi_i$. As demonstrated in \ifArxiv Appendix~\ref{ap:cor_2}\else Appendix E of~\cite{arxiv_version}\fi, once $A_i$ enters the blocking mode, $\phi_i$ converges to the bearing angle $\beta_i^j$, for all $i,j\!\in\!\{1,2\}$ with $i\!\neq\!j$. 
When airplane $A_i$ first reaches the condition $\phi_i\!=\!\beta_i^j$ before the other, the blocking mode transitions to the avoiding mode at the next moment, as the blocking condition is no longer satisfied. This self-unblocking behavior is shown in Fig.~\ref{fig:blocking_curve} and formalized in the following corollary.

\begin{corollary}[Self-unblocking condition]
 Considering the system setting described in Theorem~\ref{th:block} and an airplane is in blocking mode, the blocking mode ends when $\exists i,j \in \{1,2\}$ and $i \neq j$ such that $\phi_i = \beta_i^j$ and $\phi_j \neq \beta_j^i$. 
\label{th:unblock}
\end{corollary}
\begin{proof}  \renewcommand{\qedsymbol}{}
\ifArxiv The proof is provided in Appendix~\ref{ap:cor_2}.
\else The proof is provided in Appendix E of~\cite{arxiv_version}.\fi 
\end{proof}

\noindent Geometrically, $\phi_i\!=\!\beta_i^j$ implies that $\mathbf{p}_i$, $\mathbf{p}_j$ and $\mathbf{T}_i$ are collinear, as depicted by $l_e$ in Fig.~\ref{fig:blocking_traj}. After that moment, the blocking will be broken. Using this property, we can calculate the upper and lower bounds of blocking mode duration, respectively denoted by $T_{lb}$ and $T_{ub}$, based on the simple geometric relationship, where
\begin{equation}\label{eq:time}
\begin{aligned}
T_{lb} &  \textstyle =  \frac{ \min(\|\mathbf{p}_1-\mathbf{T}_1\| \text{cos}(\beta_1^2-\phi_1), \|\mathbf{p}_2-\mathbf{T}_2\| \text{cos}(\beta_2^1-\phi_2))}{v} \\
T_{ub} &  \textstyle = T_{lb} + \frac{\|\mathbf{p}_1 - \mathbf{p}_2\| - r}{2v}.
\end{aligned}
\end{equation}
\ifArxiv The rigorous derivation is provided in Appendix~\ref{ap:time}. \else The derivation is provided in Appendix F of~\cite{arxiv_version}. \fi

\section{Consequences of Blocking Mode}

Besides the finite-time blocking modes, deadlock and livelock phenomena can occur under specific conditions. In this section, we qualitatively outline the conditions that lead to deadlock and livelock and explain why blocking is more common than deadlock.

\subsection{Deadlock and livelock phenomenon}~\label{ap:dead_live}
As stated in Theoreom~\ref{th:sol}, when $\phi_i = \beta_i^j$, both $\beta_i^j\! -\! \Delta$ and $\beta_i^j\! +\! \Delta$ are feasible solutions, for $i,j \!\in\! \{1,2\}$ and $i\!\neq\!j$. To resolve this ambiguity, we introduce a preferred direction for each airplane $A_i$, encoded by $\lambda_i\!\in\!\{-1,1\}$. This preference specifies that the heading angle is chosen as $\theta_i \!=\! \beta_i^j \!+\! \lambda_i \Delta$ when $\phi_i \!=\! \beta_i^j$. In the subsequent discussion of potential deadlock and livelock, this preferred direction will play a crucial role.

\begin{figure}[ht]
  \centering
  \begin{subfigure}[b]{0.23\textwidth}
    \centering
    \includegraphics[height=4.4cm]{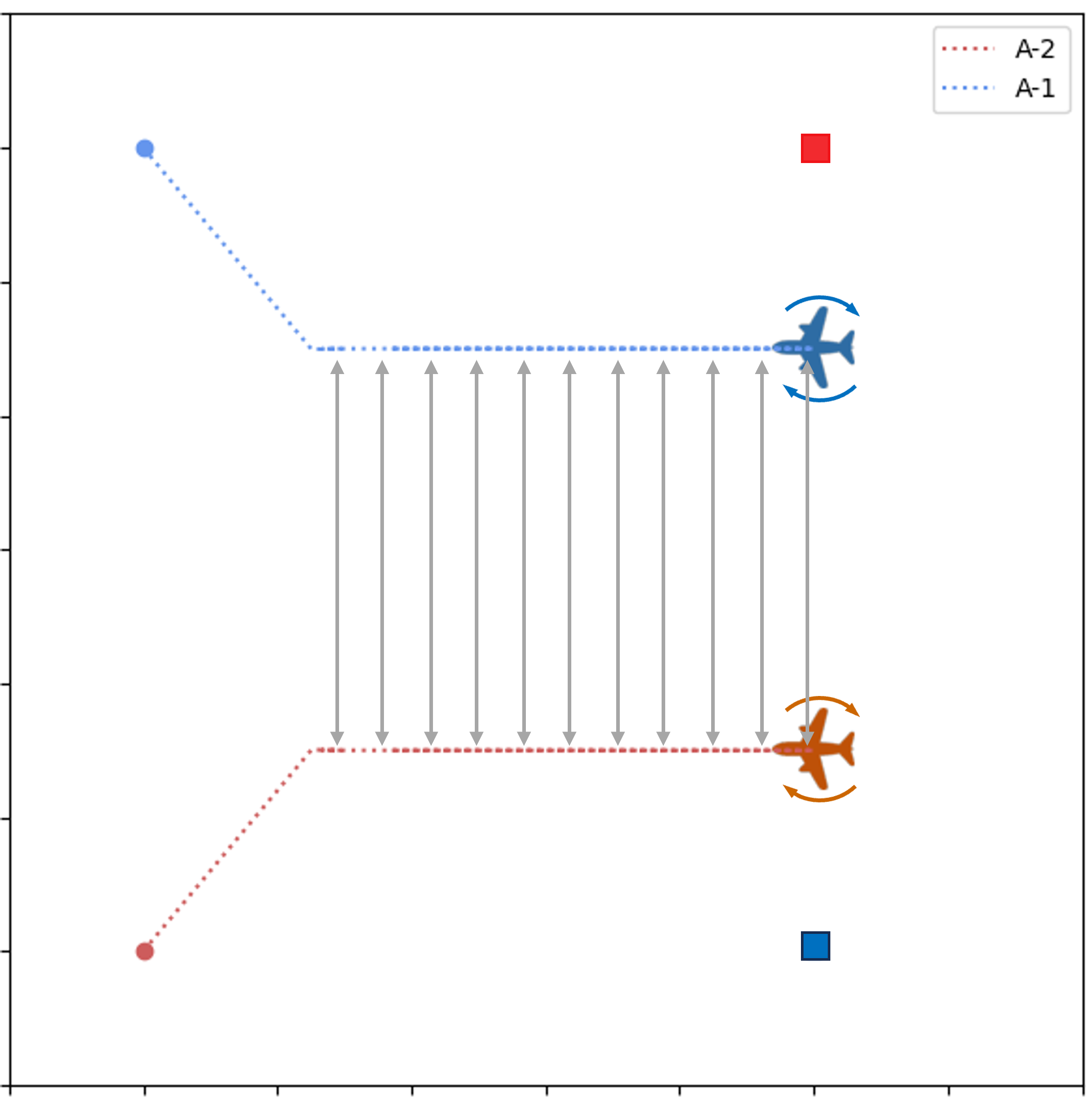}
    \caption{Deadlock: Both airplanes repeatedly reverse their headings but remain fixed in place.}
    \label{fig:deadlock}
  \end{subfigure}\hfill
  \begin{subfigure}[b]{0.23\textwidth}
    \centering
    \includegraphics[height=4.4cm]{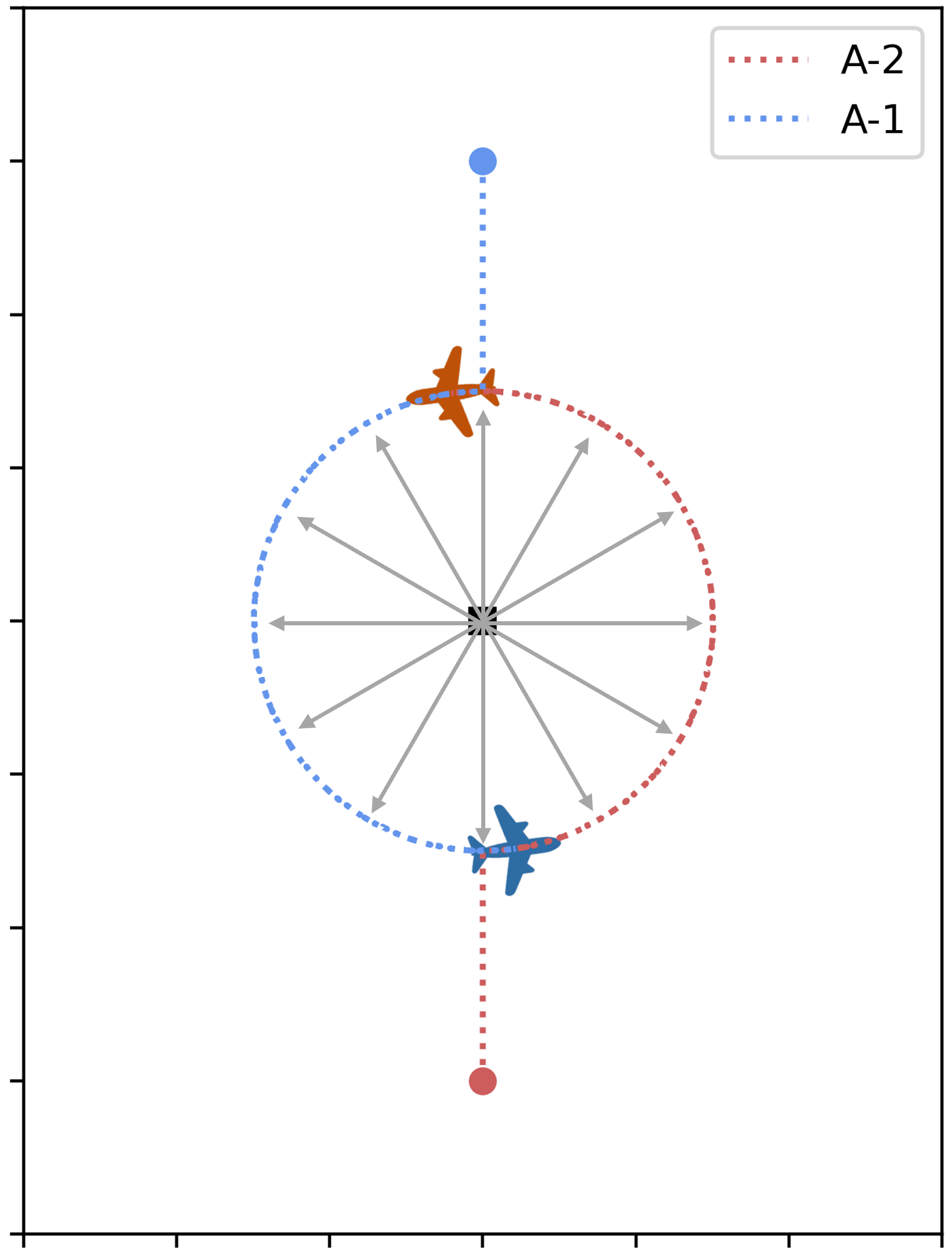}
    \caption{Livelock: Both airplanes circle around their shared target without ever reaching it.}
    \label{fig:livelock}
  \end{subfigure}
  \caption{Deadlock and livelock phenomenon in simulation.}
  \label{fig:dead_live_lock}
\end{figure}

If the preferred actions of airplanes are opposite, expressed as $\lambda_1\!=\!-\lambda_2$, blocking may deteriorate into deadlock, potentially resulting in a mid-air crash. In particular, deadlock occurs when $\phi_1 \!=\! \beta_1^2$ and $\phi_2 \! = \! \beta_2^1$ always hold simultaneously such that the self-unblocking condition in Corollary~\ref{th:unblock} is never satisfied. As shown in Fig.~\ref{fig:deadlock}, blocking mode can deteriorate into a deadlock when two additional conditions are met: 1) the relative position of airplanes is negatively proportional to their target displacement, expressed as $\mathbf{p}_1-\mathbf{p}_2 \!=\! -k (\mathbf{T}_1-\mathbf{T}_2)$, where $k\!>\!0$; 2) $\lambda_1 \!=\! -\lambda_2$. Due to the first condition, both airplanes arrive at the line $l_t$ simultaneously, resulting in the collinearity of $\mathbf{p}_1$, $\mathbf{p}_2$, $\mathbf{T}_1$, and $\mathbf{T}_2$, \new{which also matches the deadlock condition for mobile robots as described in~\cite[Theorem 1]{grover2023before}}. Since the collinearity implies that $\phi_1\!=\!\beta_1^2$ and $\phi_2\!=\!\beta_2^1$, the second condition further causes both airplanes to switch their avoiding directions simultaneously and perpetually maintain this configuration, thereby remaining trapped in a deadlock. Therefore, blocking can deteriorate into deadlock, while resolving the blocking incidentally eliminates deadlocks as well.

\new{If $\lambda_1 \!=\! \lambda_2$, deadlock does not occur. However, a livelock phenomenon may arise when $\lambda_1 \! = \!\lambda_2$ if both airplanes share the same target position.} As depicted in Fig.~\ref{fig:livelock}, airplanes fly around the target endlessly without reaching it when $\mathbf{T}_1\!=\!\mathbf{T}_2$ and $\lambda_1 \!=\! \lambda_2$. In this case, continuous rotation is ineffective in resolving conflicts. Nevertheless, this scenario represents a special type of encounter that typically occurs near airports. Since this paper focuses on mid-air encounters, we assume $\|\mathbf{T}_1 \! - \! \mathbf{T}_2\| \! \geq \! r$ throughout, thereby ruling out livelock.

\subsection{Comparison of blocking and deadlock}~\label{sec:dead_live}
\begin{figure}[th]
    \centering \vspace{-3mm}
\includegraphics[width=0.32\textwidth]{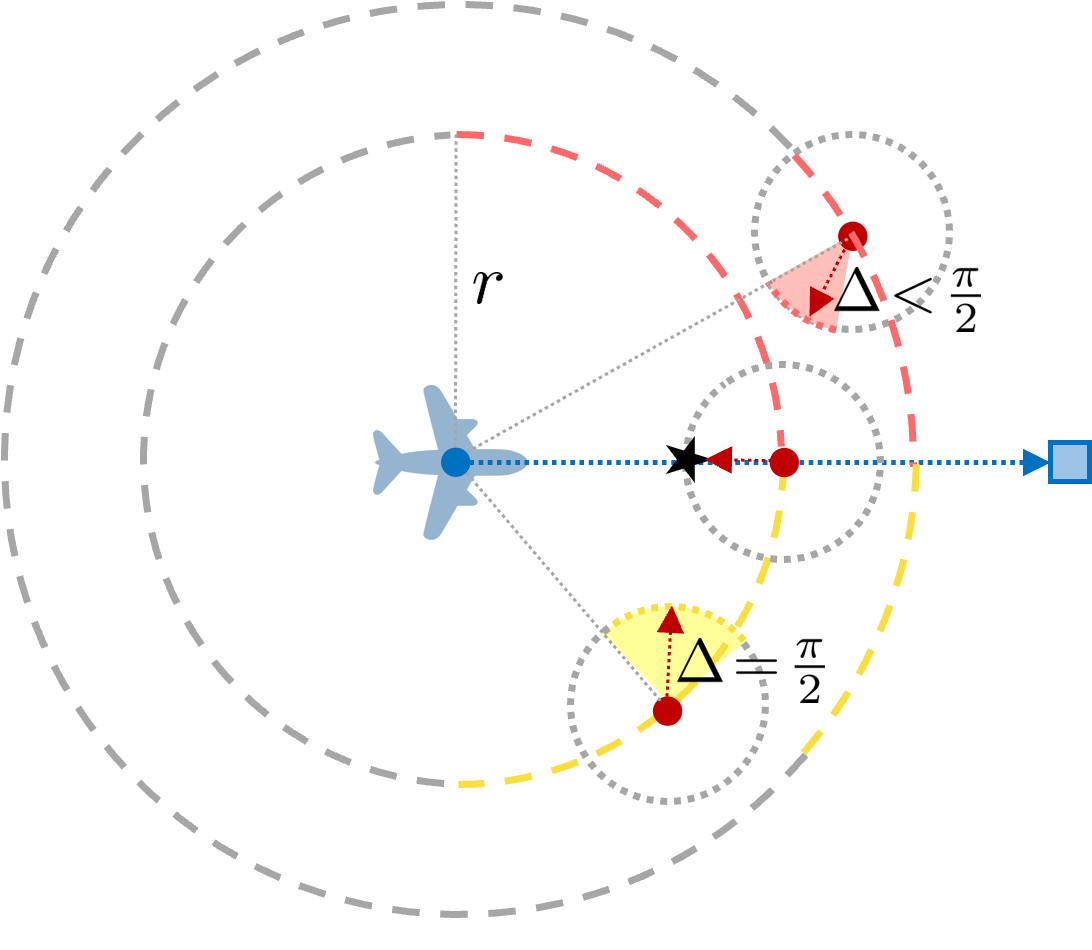}
    \caption{Visualization of blocking and deadlock conditions.}
    \label{fig:condition_compare} 
\end{figure}

To compare the restrictiveness of the deadlock and blocking conditions, Fig.~\ref{fig:condition_compare} illustrates both conditions from the perspective of the ego (blue) airplane. Given an ego airplane flying towards a target, the blocking and deadlock conditions depend on the position and desired cruising angle of the other airplane. The blocking is triggered if the other airplane's position lies along the colored curves and its desired cruising direction falls within the corresponding colored arcs. In terms of~\cite[Theorem 1]{grover2023before}, deadlock for mobile robots is induced when two agents approach each other head-on, as indicated by the star marker. Let us analyze the likelihood of blocking versus deadlock. Assuming airplanes with positions uniformly distributed along a circle with radius $r$ and uniformly distributed $\phi_j$, the probability of a blocking event is $\frac{1}{8}$, while the probability of deadlock is 0. Therefore, airplanes face a greater risk of blocking than of deadlock in an encounter.

\section{Blocking Resolution Framework}~\label{sec:unblock}
In this section, we aim to show that without central coordinators and without communication, airplanes can still resolve blocking modes by collecting enough information to make adaptive decisions instead of following a fixed rule. First, we propose an unblocking maneuver designed to break blocking modes, along with an adaptive-priority rule that selects the airplane to execute the unblocking maneuver based on a duration assessment. Here, it is still assumed that both agents know each other's target positions. Therefore, we also design an interactive maneuver that lets each airplane reveal its intentions and estimate the other's intentions without relying on direct communication.

\subsection{Unblocking maneuver}
The analysis in Sec.~\ref{sec:condition} revealed that both the occurrence and duration of blocking modes depend largely on the cruising angle. Based on the self-unblocking condition in Corollary~\ref{th:unblock}, we propose an unblocking maneuver that temporarily relocates the target point to adjust the cruising angle, together with an adaptive-priority rule that selects which airplane executes the maneuver. This is formalized in Alg.~\ref{alg:unblock} for airplane $A_i$, where $\hat{\mathbf{T}}_i$ denotes the target currently being pursued by $A_i$. For clarity, let $A_i$ denote the ego airplane applying the resolution strategy, and $A_j$ the opponent airplane, where $i, j \!\in\! \{1, 2\}$ and $i \!\neq \!j$. For unblocking, the location of $\hat{\mathbf{T}}_i$ can switch between the desired target $\mathbf{T}_i$ and a temporary target.  In Alg.~\ref{alg:unblock}, when a blocking mode is detected, $A_i$ can unblock by temporarily redirecting its target to the current position $\mathbf{p}_j$ of $A_j$ (line~\ref{step:unblocking}). At the next moment, the self-unblock condition for $A_i$ is triggered instantaneously, and the safety filter immediately steers $A_i$ back in the opposite direction to bypass $A_j$. Once $A_i$ reaches the temporary target (line~\ref{step:reach}), it resumes cruising mode toward its original target point $\mathbf{T}_i$. Since the unblocking maneuver only adjusts the cruising angle within the safety filter~\eqref{eq:sf}, the resulting control input remains provably safe.

However, if both airplanes initiate the unblocking maneuver simultaneously, a new blocking situation may arise. To prevent this, a priority mechanism is required to determine which airplane should perform the unblocking maneuver. Specifically, we implement an adaptive priority mechanism that selects the option with the shortest estimated duration among three options: maintaining the current blocking mode, unblocking by the ego airplane $A_i$, and unblocking by the opponent airplane $A_j$. The total durations required for the airplanes to reach their respective targets under each option are denoted by $T_b$, $T_u^i$, and $T_u^j$, respectively. The details for approximating these durations are provided in Appendix~\ref{ap:priority} based on the estimated blocking time~\eqref{eq:time}. When all airplanes apply the same duration estimation method, they can independently reach a consistent priority decision without direct communication. However, some strictly symmetric situations may lead to $T_u^i \!=\! T_u^j$. In these rare cases, a fixed priority, such as the \textit{right-hand rule}~\cite{pierson2020weighted}, is used to break the symmetry. Compared to applying fixed priorities in all situations, the adaptive priority mechanism reduces inefficiencies and avoids unnecessary energy consumption.

\RestyleAlgo{ruled}
\SetNoFillComment
\SetKwComment{Comment}{//}{}
\SetKwProg{Fn}{Function}{:}{end}
\begin{algorithm}[hbt!]
\caption{Unblocking Manuever for $A_i$}\label{alg:unblock}
\textbf{Initialize:} $\hat{\mathbf{T}}_i \leftarrow \mathbf{T}_i$, $\mathbf{T}_j$\;
\While{$\mathbf{p}_i \neq \mathbf{T}_i$}{
    \textbf{Input:} updated $\mathbf{p}_i$, $\mathbf{p}_j$, $\theta_j$\;
    $\tilde{\mathbf{u}}_i \leftarrow \mathcal{N}(\hat{\mathbf{T}}_i, \mathbf{p}_i)$ in Eq.~\eqref{eq:u_n}\;
    $\mathbf{u}_i \leftarrow \mathcal{F}(\mathbf{p}_i, \mathbf{p}_j, \tilde{\mathbf{u}}_i)$ in Eq.~\eqref{eq:sf}\;

    \If{$A_i$ is in blocking mode (c.f., Def.~\ref{def:blocking})}{
        Compute $T_b$, $T_u^i$ and $T_u^j$\;
            \If{$T_u^i$ = min($T_b$, $T_u^i$, $T_u^j$)}{ 
                $\hat{\mathbf{T}}_i \leftarrow \mathbf{p}_j$  \Comment*[l]{\small Set temporary target}\label{step:unblocking}
            }
    }

    \If{$\mathbf{p}_i = \hat{\mathbf{T}}_i$ and $\hat{\mathbf{T}}_i \neq \mathbf{T}_i$\label{step:reach}}{
        $\hat{\mathbf{T}}_i \leftarrow \mathbf{T}_i$ \Comment*[l]{\small Restore original target}
    }
    \Return $\mathbf{u}_i$\;
}
\end{algorithm}

\subsection{Interactive maneuver for intention estimate}
To collect sufficient information to estimate the target point of the other airplane without communication, the ego airplane can actively provoke the opponent to reveal its intentions. To achieve this, we design an interactive maneuver in which the ego airplane temporarily moves away from the opponent until the opponent’s safety filter deactivates, causing the opponent to revert to its cruising mode and naturally head toward its target. Specifically, as shown in line~\ref{step:interactive} of Alg.~\ref{alg:estimate}, when the airplane $A_i$ initiates an interactive maneuver, the control input $\mathbf{u}_i$ produced by the safety filter is modified by adding $\Delta \mathbf{u} \!=\! k(\mathbf{p}_i - \mathbf{p}_j)$, where $k\!>\!0$. The modified control inputs are provably safe. Regardless of whether one or both airplanes perform this maneuver, both airplanes are driven to revert to cruising mode and reveal their intended heading directions. In the absence of uncertainty, observing at least two distinct poses during these maneuvers is sufficient to estimate the target point using a classical triangulation algorithm~\cite{acc2012triangulation}. Specifically, let $\tilde{\mathbf{T}}_i^j$ represent the target of $A_j$ as estimated by $A_i$, and a set $\mathcal{T}_i^j$ is maintained by $A_i$ to collect poses of $A_j$ (line~\ref{step:append} in Alg.~\ref{alg:estimate}). Once $\mathcal{T}_i^j$ contains a pair of non-collinear poses, denoted by $(\mathbf{p}_j^{1}, \theta_j^{1})$ and $(\mathbf{p}_j^{2}, \theta_j^{2})$, the intersection of the corresponding heading directions is used to estimate the target point. This estimate satisfies 
\begin{equation}~\label{eq:est_target}
     \tilde{\mathbf{T}}_i^j =\mathbf{p}_j^{1} + k_1 \mathbf{d}^{1} = \mathbf{p}_j^{2} + k_2 \mathbf{d}^{2},
\end{equation}
where $\mathbf{d}^{1} \!=\! (\cos\theta_j^{1}, \sin\theta_j^{1})$ and $\mathbf{d}^{2} \!=\! (\cos\theta_j^{2}, \sin\theta_j^{2})$ are the observed cruising directions, and $k_1, k_2 \in \mathbb{R}_{>0}$ are computed accordingly. When more poses are available, the triangulation algorithm in \cite{acc2012triangulation} can be applied for improved accuracy in uncertain environments. By incorporating this estimation procedure into line 6 of Alg.~\ref{alg:unblock}, the resolution strategy becomes applicable even when the target of $A_j$ is initially unknown. Note that this resolution strategy not only resolves the blocking in the control framework defined in Sec.~\ref{sec:model_blocking}, but also holds potential for other control frameworks.

\begin{algorithm}[h]
  \setlength{\topsep}{0pt}
  \setlength{\partopsep}{0pt}
\caption{Interactive Maneuver for $A_i$}\label{alg:estimate}
\textbf{Initialize:} $\tilde{\mathbf{T}}_i^j \leftarrow \varnothing$, $\mathcal{T}_i^j \leftarrow \emptyset$\;

    \textbf{Input:} updated $\mathbf{p}_i$, $\mathbf{p}_j$, $\theta_j$\;

    \If{$A_i$ in blocking mode (c.f., Def.~\ref{def:blocking}) and $\tilde{\mathbf{T}}_i^j \!=\! \varnothing$}{
    $\mathbf{u}_i \leftarrow \frac{\mathbf{u}_i + \Delta\mathbf{u}}{\|\mathbf{u}_i + \Delta\mathbf{u} \|}$\label{step:interactive}\Comment*[l]{\small Fly away from $A_j$}
    \Return $\mathbf{u}_i$\;
    }
    \If{$A_i$ is in cruising mode}{Append $(\mathbf{p}_j, \theta_j)$ to $\mathcal{T}_i^j$~\label{step:append}\Comment*[l]{\small Collect poses}}
    \If{$\exists (\mathbf{p}_j^{1}, \theta_j^{1}), (\mathbf{p}_j^{2}, \theta_j^{2}) \!\in\! \mathcal{T}_i^j$ with $\theta_j^{1} \! \neq \! \theta_j^{2}$}
    {     
        $\tilde{\mathbf{T}}_i^j \leftarrow$ Sloving Eq.~\eqref{eq:est_target}\Comment*[l]{\small Estimate $\mathbf{T}_j$} 
    }
\end{algorithm}

\section{Simulation experiments}~\label{sec:sim}
This section validates the efficacy and utility of the proposed resolution strategy. The video demonstration is accessible at the link\footnote{{\small\url{https://youtu.be/r8v7qULA3fM}}}.


\subsection{Simulation for the resolution strategy}
\begin{figure}[th]
    \centering \vspace{-3mm}
    \includegraphics[width=0.4\textwidth]{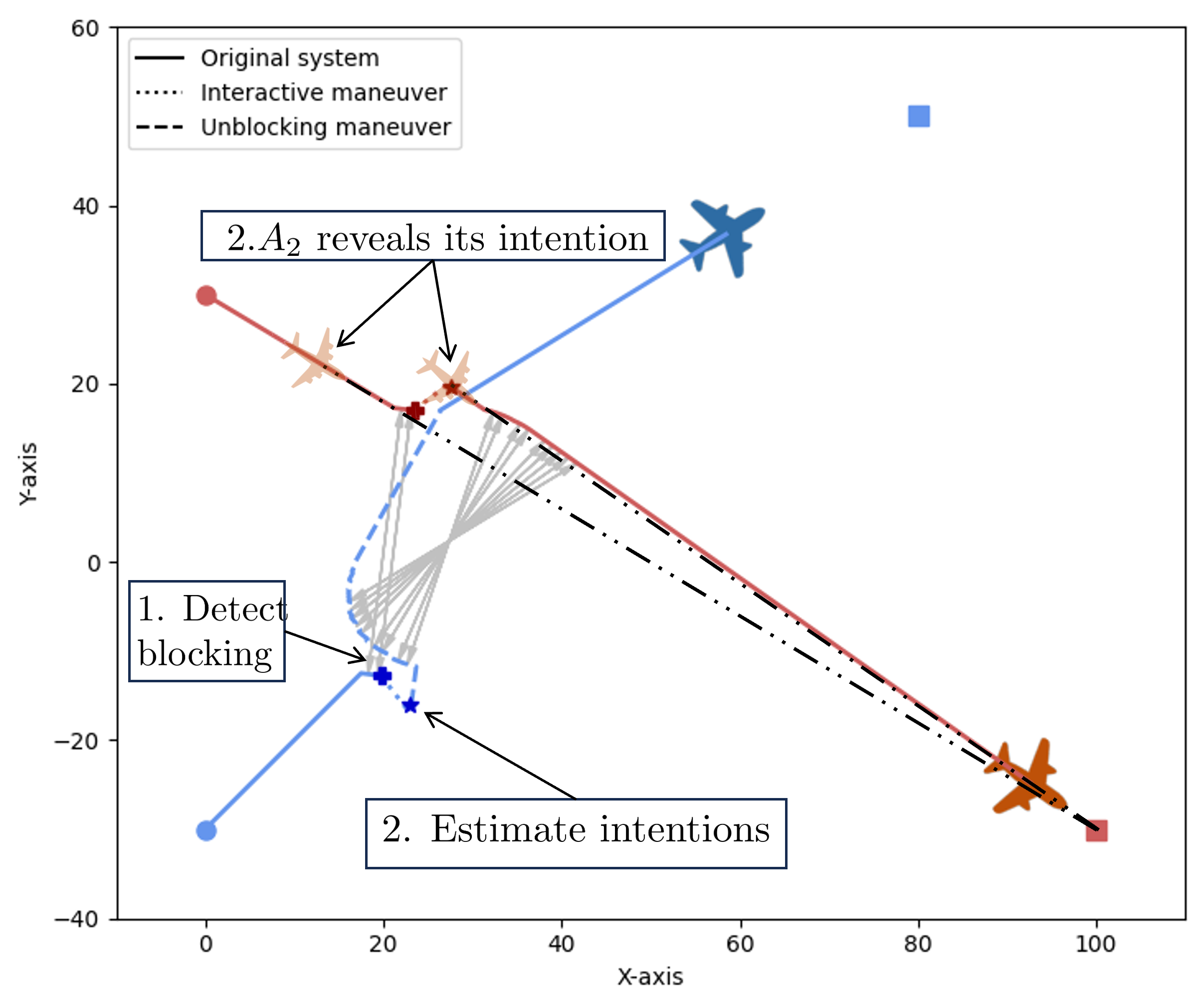}
    \caption{Blocking resolution process using Alg.~\ref{alg:unblock}. Plus markers represent the moment when the blocking mode is detected; star markers represent the moment of successful intention estimation. The intersection point of two black dotted dashed lines is the estimated target of $A_2$ from $A_1$.  
    } 
    \label{fig:unblock_sim} 
\end{figure}

The simulation in Fig.~\ref{fig:unblock_sim} illustrates the unblocking process when two airplanes apply the resolution strategy introduced in Sec.~\ref{sec:unblock}. Airplanes $A_1$ and $A_2$ are initially positioned at $(0, -30)$ and $(0, 30)$, respectively, and are directed toward their target destinations at $\mathbf{T}_1\!=\! (80, 50)$ and $\mathbf{T}_2\! =\!(100, -30)$. The system parameters are set to $v\!=\!5$, $r\!=\!30$, and $\alpha\!=\!3$. First, the two airplanes initiate interactive maneuvers once a blocking mode is detected. Then, as they revert to cruising mode, their intentions are revealed, enabling mutual estimation of each other’s target points. This demonstrates that the intentions of other airplanes can be successfully estimated using their reactions without any communication. Using the estimated target points, $A_1$ is selected to perform an unblocking maneuver based on the time assessment. The original target $\mathbf{T}_1$ is replaced by a temporary target point that can break the ongoing blocking mode and immediately make airplanes revert to an avoiding mode. When $A_1$ reaches the temporary target, the encounter conflict is resolved at this moment, and $A_1$ heads toward the original target. Note that the simulation is based on the single integrator model~\eqref{eq:sys}. To test the strategy on the unicycle model~\eqref{eq:uni_sys} together with an angle-tracking controller,  Fig.~\ref{fig:unblock_large_scale} simulates four airplanes flying toward their respective targets, each sequentially encountering two different airplanes along the way. Each encounter is resolved smoothly, validating the efficacy of the proposed resolution strategy.

\begin{figure}[th]
    \centering 
    \includegraphics[width=0.4\textwidth]{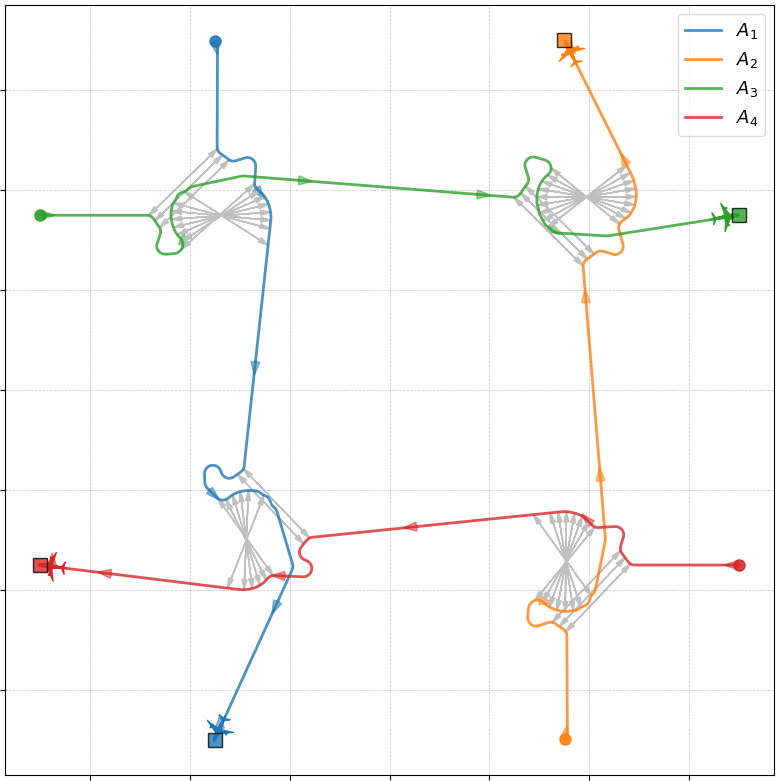}
    \caption{Encounter simulations involving four airplanes. Circle markers indicate initial positions, and square markers denote target points.} 
    \label{fig:unblock_large_scale}\vspace{-5mm}
\end{figure}

\subsection{Monte Carlo simulations}
To evaluate the adaptive priority, we test the proposed strategy in randomly generated encounter scenarios. In these tests, the initial positions of two airplanes are fixed at a distance equal to the safety margin, while their target positions are randomly selected within regions prone to blocking. In 100 randomly generated encounter scenarios, we measure the task completion time under three strategies: maintaining the blocking mode, unblocking with fixed priority, and unblocking with adaptive priority. Compared with maintaining blocking, the fixed-priority strategy reduces average flight time by $16.7\%$, while the adaptive-priority strategy achieves a $21.3\%$ reduction. The benefit of adaptive priority becomes even more pronounced for longer journeys. Undoubtedly, adaptive priority statistically improves the air-traffic efficiency. Importantly, the blocking mode can be resolved quickly in each case, and all safety constraints are maintained throughout the simulations.

\section{Conclusion and Future Work}
Beyond the widely studied deadlock phenomenon, this paper investigates the finite-time blocking phenomenon in a two-airplane encounter scenario. Our analytical results reveal the underlying conditions under which blocking occurs, evaluate its impact, and demonstrate that the blocking condition is much less restrictive than that of deadlock. To address this issue, we develop an intention-aware resolution strategy that enables adaptive decision-making without relying on communication or central coordination, making it well-suited as a fallback in harsh situations. 

While blocking behavior may vary across systems, finite-time blocking is likely to occur in other multi-agent scenarios. This work provides an initial study toward understanding such phenomena. In the future, we will account for more realistic three-dimensional airplane dynamics, rather than a simplified horizontal model. Although multi-airplane encounters are less common, we will study these challenging scenarios, which require decision-making more sophisticated than a simple left/right bypass. Lastly, we will propose a blocking-free controller that incorporates both unblocking and interactive maneuvers without relying on the protocol-based framework.


\normalem
\bibliographystyle{ieeetr}
\bibliography{ref}

\ifArxiv
\appendix

\subsection{Blocking induced by different controllers}\label{ap:vo}
Aside from the CBF-based safety filter, two widely used collision‑avoidance methods, velocity obstacles and potential fields, also produce similar blocking behaviors under symmetry, showing the blocking phenomenon is general in avoiding controllers.

\textit{1) Velocity obstacle} for $A_i$ is the set of all velocities for $A_i$ that will result in a collision with $A_j$ over a time horizon $\tau$, assuming $A_j$ maintains its current velocity:
$$
\mathrm{VO}_{i|j} = \textstyle \left\{ \mathbf{u}_i \mid \exists t \in [0, \tau], \; \|\mathbf{p}_j + \mathbf{u}_j t - (\mathbf{p}_i + \mathbf{u}_i t)\| \le r \right\},
$$
where  $i,j \!=\! \{1,2\}$ and $i\! \neq\! j$. Like the safety filter, velocity command is normally selected to minimize deviation from the preferred velocity~\cite{van2008reciprocal} outside the velocity obstacle, $$ {\operatorname{min}} \, \frac{1}{2}\|\mathbf{u}_i-\mathbf{\tilde{u}}_i \|^2 \, \text{s.t.} \, \mathbf{u}_i \notin \mathrm{VO}_{i|j}  $$
In this case, the same parallel flying phenomenon also happens as shown in Fig.~\ref{fig:blocking_traj}. 

\textit{2) Potential Field} treats each agent as a particle moving in an artificial potential field. The target exerts an attractive force, $F_{att}$, that pulls the robot toward the target, while the other agent exerts a repulsive force, $F_{rep}$, that pushes the agent away. Both forces depend on the inter‑agent distance. The heading direction $\theta_i$ align with the resultant force $F^i = F_{att}^i + F_{rep}^i$. Although this method does not minimize deviation from a preferred velocity, it also exhibits the parallel‑flying behavior under symmetric conditions.

\subsection{Proof of Theorem~\ref{th:sol}}~\label{ap:thm_1}
\begin{figure}[htb]
    \centering \vspace{-3mm}
    \includegraphics[width=0.27\textwidth]{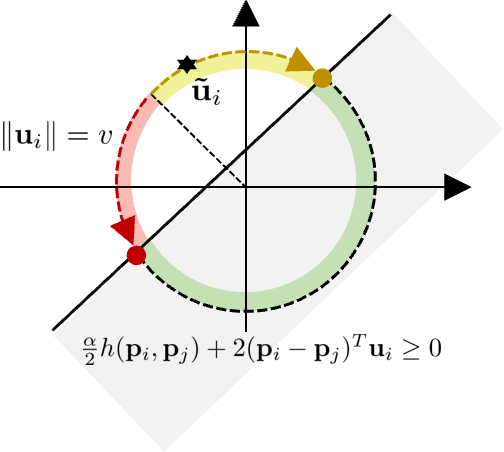}
    \caption{Geometric interpretation of the safety filter~\eqref{eq:sf}.} \vspace{-3mm}
    \label{fig:geo_int} 
\end{figure}

Fig.~\ref{fig:geo_int} provides a geometric interpretation of the safety filter in Eq.~\eqref{eq:sf}. The dashed circle represents the constant-speed constraint from Eq.~\eqref{sf:norm}, whereas the shaded half-plane depicts the CBF inequality in Eq.~\eqref{sf:cbf}. The green arc represents the set of safe inputs, while the yellow and red arcs denote the unsafe sets. The objective in Eq.~\eqref{sf:obj} selects a safe input closest to the cruising input. Hence, if the cruising input already lies on the green arc, it is accepted unchanged; otherwise, the filter corrects it into the nearest safe point. Specifically, a cruising input on the yellow arc is shifted to the adjacent yellow circle marker, whereas one on the red arc is moved to the red circle marker.

To numerically analyze the optimization problem in Eq.~\eqref{eq:sf}, it is better to use 
the parametric approach of $\mathbf{u}_i\!=\!v[\text{cos}(\theta_i), \text{sin}(\theta_i)]^T$ and $\mathbf{\tilde{u}}_i\!=\!v [\text{cos}(\phi_i), \text{sin}(\phi_i)]^T$, which  implicitly satisfy the norm constraint in Eq.~\eqref{sf:norm}. The original objective in Eq.~\eqref{sf:obj} can be derived as,
\begin{align*}
    \frac{1}{2}\|\mathbf{u}_i-\mathbf{\tilde{u}}_i \|^2 & \textstyle = 
    \frac{1}{2} \left\|\begin{bmatrix}
        \text{cos}(\theta_i) -  \text{cos}(\phi_i) \\
        \text{sin}(\theta_i) -  \text{sin}(\phi_i)
    \end{bmatrix} \right\|^2 \\
    & = 1 - cos(\theta_i - \phi_i)
\end{align*}
The CBF condition in Eq.~\eqref{sf:cbf} can be derived into, 
\begin{align*}
     & \textstyle \frac{\alpha}{2} h(\mathbf{p}_i, \mathbf{p}_j) - 2v \|\mathbf{p}_i - \mathbf{p}_j\| 
    \begin{bmatrix}
        \text{cos}(\beta_i^j) \\ \text{sin}(\beta_i^j)
    \end{bmatrix}^T
    \begin{bmatrix}
    \text{cos}(\theta_i) \\ \text{sin}(\theta_i)
    \end{bmatrix}  \geq 0 \\
    & \textstyle \Rightarrow  cos(\theta_i-\beta_i^j) \leq \frac{\alpha h(\mathbf{p}_i, \mathbf{p}_j)}{4v \|\mathbf{p}_1 - \mathbf{p}_2\|}
\end{align*}
Let $L$ denote the right-hand term of the inequality function. Since $\alpha$ and $v$ are positive constant numbers, $L \geq 0$ if the current state is in the safe set. Given the above derivation, the optimization problem in Eq.~\eqref{eq:sf} can be transformed into the formulation from the angular perspective, 
\begin{equation}
    \begin{aligned}
     & \min_{\theta \in (-\pi, \pi]}  &&  1 - cos(\theta_i - \phi_i) \\
    & \quad \text{s.t.} &&  cos(\theta_i -\beta_i^j) \leq L
    \end{aligned}
    \label{eq:sf_2}
\end{equation}
Without the loss of generality, we assume $\beta_i^j, \theta_i, \phi_i \in (-\pi, \pi]$. With the angular normalization operator $\measuredangle(\cdot) \in (-\pi, \pi]$, the optimization can be further transformed into,
\begin{equation}
    \begin{aligned}
     & \min_{\theta_i \in (-\pi, \pi]}  && |\measuredangle(\theta_i - \phi_i)| \\
    & \quad \text{s.t.} && |\measuredangle(\theta_i - \beta_i^j)| \in  [\Delta, \pi]
    \end{aligned}
    \label{eq:sf_3}
\end{equation}
where $\Delta = \operatorname{cos}^{-1}(\min(1, L) \in [0, \frac{\pi}{2}]$. With the reformulated optimization problem from an angular perspective, the explicit solution can be easily obtained in the form of piecewise functions, as shown in Eq.~\eqref{eq:sol}.

\subsection{Proof of Theorem~\ref{th:block}}\label{ap:thm_2}
For clarity, the analysis below uses $A_1$ and $A_2$ instead of the generic pair $A_i$ and $A_j$, where $i,j\in{1,2}, i\neq j$.

\textbf{Sufficiency}: As for condition (1), since $s \measuredangle(\phi_1-\beta_1^2) \in  [0, \Delta)$ and $-s\measuredangle(\phi_2-\beta_2^1) \in [0, \Delta)$, the heading directions corrected by safety filters are $\theta_1\!=\! \measuredangle(\beta_1^2 +s \Delta)$ and $\theta_2\!=\! \measuredangle(\beta_2^1 -s \Delta)$. Since $\measuredangle(\beta_1^2 - \beta_2^1)\!=\! \pi$, 
$$
\begin{aligned}
\mathbf{u}_2 - \mathbf{u}_1 & \textstyle = \begin{bmatrix}
    \text{cos}(\beta_1^2 + s \Delta) - \text{cos}(\beta_2^1 -s \Delta) \\
    \text{sin}(\beta_1^2 +s \Delta) - \text{sin}(\beta_2^1 -s \Delta)
\end{bmatrix} \\ & = 2 \text{cos}(\Delta) \begin{bmatrix}
    \text{cos}(\beta_1^2) \\
    \text{sin}(\beta_1^2)
\end{bmatrix} \\
& = 2 \frac{\text{cos}(\Delta)}{d_{1,2}} (\mathbf{p}_2-\mathbf{p}_1)
\end{aligned}.
$$
Thus, $(\mathbf{u}_2 - \mathbf{u}_1)$ is collinear to  $(\mathbf{p}_2 - \mathbf{p}_1)$ such that $\dot{\beta}_i^j\!=\! 0$. Therefore, $s\measuredangle(\phi_1-\beta_1^2) \in [0, \Delta)$ and $-s\measuredangle(\phi_2-\beta_2^1) \in [0, \Delta)$ is a sufficient condition for both airplanes in blocking mode. As for condition (2),  $s\measuredangle(\phi_1-\beta_1^2) \in [0, \Delta)$ implies that $\theta_1\!=\! \measuredangle(\beta_1^2 + s \Delta)$. Given $\phi_2\!=\! \measuredangle(\beta_1^2 \pm \Delta)$, $\mathbf{u}_1 - \mathbf{u}_2\!=\! 0$ such that $\dot{\beta}_i^j\!=\! 0$. In addition, since $\measuredangle(\beta_1^2 - \beta_2^1)\!=\! \pi$, $|\measuredangle(\phi_2^1-\beta_2^1)| \notin [0, \Delta)$, which means $A_2$ is in cruising mode. Therefore, the sufficiency of condition (2) is verified. 

\textbf{Necessity}: Suppose $A_1$ is in blocking mode such that $\theta_1\!=\! \measuredangle(\beta_1^2 + s \Delta)$. To make $\dot{\beta}_1^2\!=\!\dot{\beta}_2^1\!=\! 0$, the following equality should be satisfied,
$$ \textstyle \begin{aligned}
    & (\mathbf{p}_2\!-\! \mathbf{p}_1) \times (\mathbf{u}_2\!-\! \mathbf{u}_1) \\
    = & (\mathbf{p}_2\!-\! \mathbf{p}_1) \times \mathbf{u}_2\!-\! (\mathbf{p}_2\!-\! \mathbf{p}_1) \times \mathbf{u}_1 \\
    = & v \|\mathbf{p}_2\!-\! \mathbf{p}_1\| \text{sin}(\beta_1^2-\theta_2) \vec{\mathbf{n}}\!-\! \|\mathbf{p}_2\!-\! v \mathbf{p}_1\| \text{sin}(\beta_1^2-\theta_1) \vec{\mathbf{n}} \\ = &0
\end{aligned},
$$
where $\vec{\mathbf{n}}$ denotes a unit vector perpendicular to the horizontal plane. Therefore, 
$\text{sin}(\beta_1^2-\theta_2)\!=\! \text{sin}(\beta_1^2\!-\! \theta_1)$ to ensure $\dot{\beta}_i^j\!=\! 0$. Given $\theta_1\!=\! \measuredangle(\beta_1^2 + s \Delta)$ and $\measuredangle(\beta_1^2\!-\! \beta_2^1)\!=\! \pi$, $\text{sin}(\beta_1^2-\theta_2)\!=\! \text{sin}(\beta_1^2\!-\! \theta_1)\!=\! \text{sin}(-s \Delta)$. Therefore, $\theta_2\!=\! \measuredangle(\beta_1^2 + s \Delta)$ or $\theta_2\!=\! \measuredangle(\beta_1^2 - \pi - s\Delta)\!=\! \measuredangle(\beta_2^1 - s\Delta)$. 
\begin{itemize}
    \item[1)] Suppose $A_2$ is in blocking mode such that $\|\measuredangle(\phi_2 - \beta_2^1)\| \in [0, \Delta)$. To make $\dot{\beta}_i^j\!=\! 0$, $\theta_2\!=\! \measuredangle(\beta_2 - s\Delta)$ such that $ -s\measuredangle(\phi_2-\beta_2^1) \in [0, \Delta)$. 
    \item[2)] Suppose $A_2$ is not in blocking mode such that $\|\measuredangle(\phi_2 - \beta_2^1)\| \notin [0, \Delta)$. To make $\dot{\beta}_i^j\!=\! 0$,$\theta_2\!=\! \measuredangle(\beta_1^2 + s \Delta)\!=\! \theta_1$ such that $\phi_2\!=\! \measuredangle(\beta_1^2 + s \Delta)$.
\end{itemize}
Therefore, the necessity of conditions (1) and (2) is also verified.

\subsection{Proof of Corolloary~\ref{th:block_bf}}\label{ap:cor_1}

Under the conditions $\measuredangle(\phi_1-\beta_1^2)\!=\! -\measuredangle(\phi_2-\beta_2^1)$
and $\|\mathbf{p}_2 - \mathbf{p}_c\|\!=\! \|\mathbf{p}_1 - \mathbf{p}_c\|$, the two airplanes in cruising mode approach the encounter point $\mathbf{p}_c$ symmetrically with respect to their vertical midline. In this situation, they will activate their safety filters at the same time. Subsequently, they must enter the blocking mode because their cruising angles fall into the correction region on the same side; that is, $\exists s \in \{-1, +1\}$, $s\measuredangle(\phi_1-\beta_1^2) \in [0, \Delta)$ and $-s\measuredangle(\phi_2-\beta_2^1) \in  [0, \Delta)$.

\begin{figure*}[htb!]
  \centering
  \begin{subfigure}[b]{0.33\textwidth}
    \centering
    \includegraphics[height=3cm]{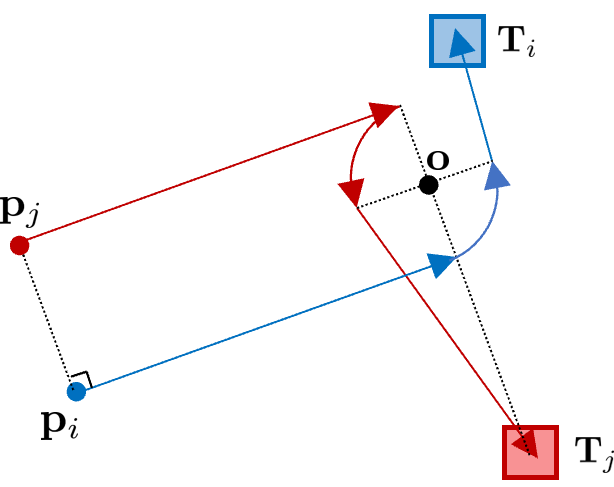}
    \caption{Maintain blocking}
    \label{fig:subfig1}
  \end{subfigure}\hfill
  \begin{subfigure}[b]{0.33\textwidth}
    \centering
    \includegraphics[height=3cm]{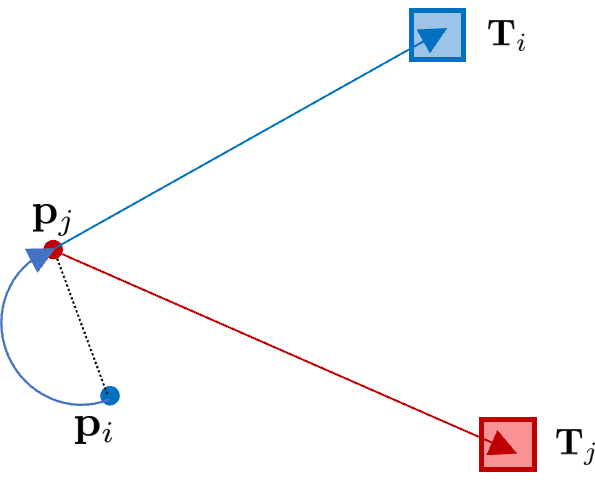}
    \caption{Unblocking by $A_i$}
    \label{fig:subfig2}
  \end{subfigure}\hfill
  \begin{subfigure}[b]{0.33\textwidth}
    \centering
    \includegraphics[height=3cm]{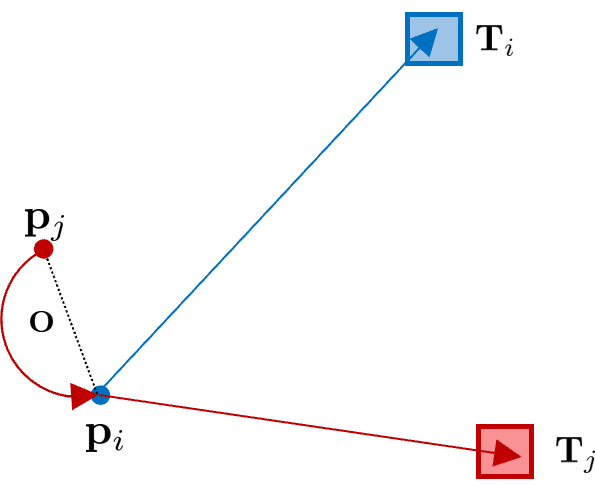}
    \caption{Unblocking by $A_j$}
    \label{fig:subfig3}
  \end{subfigure}\hfill
  \caption{Duration approximation for three options.}
  \label{fig:time_cost}\vspace{-5mm}
\end{figure*}

\subsection{Proof of Corolloary~\ref{th:unblock}}\label{ap:cor_2}

Without the loss of generality, assume $A_1$ is in blocking mode such that $\|\beta_1^2 - \phi_1\| \in [0, \Delta)$ and $\dot{\beta}_1^2\!=\!0$.
Since $\dot{\beta}_1^2\!=\!0$, $\textstyle \frac{d}{dt} \left(\textstyle \frac{\mathbf{p}_2 - \mathbf{p}_1}{\|\mathbf{p}_2 - \mathbf{p}_1\|}\right)\!=\! 0$. Given $\textstyle \text{cos}(\beta_1^2 - \phi_1)\!=\! \textstyle 
\frac{(\mathbf{p}_2 - \mathbf{p}_1)^T (T_1 - \mathbf{p}_1)}{\|\mathbf{p}_2 - \mathbf{p}_1\| \|T_1 - \mathbf{p}_1\|}
$, the derivative of $\text{cos}(\beta_1^2 - \phi_1)$ with respect to $t$ is calculated as follows,
\begin{equation}
\begin{aligned}
    & \textstyle \frac{d}{dt}\text{cos}(\beta_1^2 - \phi_1) \\
    = & \textstyle \frac{d}{dt} \left(\frac{\mathbf{p}_2 - \mathbf{p}_1}{\|\mathbf{p}_2 - \mathbf{p}_1\|}\right)^T \frac{T_1 - \mathbf{p}_1}{\|T_1 - \mathbf{p}_1\|} 
     + \frac{(\mathbf{p}_2 - \mathbf{p}_1)^T}{\|\mathbf{p}_2 - \mathbf{p}_1\|} \frac{d}{dt} \left(\frac{T_1 - \mathbf{p}_1}{\|T_1 - \mathbf{p}_1\|} \right) \\
    = & \textstyle \frac{(\mathbf{p}_2 - \mathbf{p}_1)^T}{\|\mathbf{p}_2 - \mathbf{p}_1\|} \cdot \frac{-\mathbf{u}_1 \|T_1 - \mathbf{p}_1\|^2 + (T_1 - \mathbf{p}_1) (T_1 - \mathbf{p}_1)^T \mathbf{u}_1}{\|\mathbf{T}_1 - \mathbf{p}_1\|^3}\\
    = & \textstyle \frac{-v \text{cos}(\beta_1^2-\theta_1) + v\text{cos}(\beta_1^2-\phi_1) \text{cos}(\phi_1-\theta_1)}{\|\mathbf{T}_1 - \mathbf{p}_1\|}\\
    = & \textstyle \frac{v \text{sin}(\beta_1^2 - \phi_1) \text{sin}(\phi_1 - \theta_1) }{\|T_1 - \mathbf{p}_1\|}.
\end{aligned}
\end{equation}
The above derivation used following two formulas: $\frac{d}{dt}(\mathbf{a}^T \mathbf{b})\!=\! (\frac{d}{dt}\mathbf{a})^T \mathbf{b}\!+\! \mathbf{a}^T (\frac{d}{dt} \mathbf{b})$ and $\frac{d}{dt}(\frac{1}{\|\mathbf{b}\|})\! =\! -\frac{\mathbf{b}^T}{\|\mathbf{b}\|^3}\frac{d}{dt} (\mathbf{b})$, where $\mathbf{a}$ and $\mathbf{b}$ are vectors w.r.t. $t$.
In term of the explicit solution in Theorem~\ref{th:sol}, $\|\measuredangle(\beta_1^2 - \phi_1)\|, \|\measuredangle(\phi_1-\theta_1)\| \in [0, \frac{\pi}{2}]$ and $\text{sin}(\beta_1^2 - \phi_1)\text{sin}(\phi_1-\theta_1) \geq 0$ such that $\frac{d}{dt}\text{cos}(\beta_i^j - \phi_i) \geq 0$. Since $A_1$ is in a blocking mode, $\phi_1 \neq \theta_1$ such that $\frac{d}{dt}\text{cos}(\beta_i^j - \phi_i)\!=\! 0$ if and only if $\beta_1^2\!=\! \phi_1$. Since $\|\measuredangle(\beta_1-\theta_1)\| \in [0, \frac{\pi}{2}]$, $\| \measuredangle(\beta_1^2 - \phi_1)\|$ will decrease in the blocking mode. Therefore, it is demonstrated that the desired cruising angle $\phi_i$ converges to the bearing angle $\beta_i^j$ when $A_i\!\in\!\mathcal{A}$ is in blocking mode. 

Assuming $A_i\!\in\!\mathcal{A}$ is in the blocking mode, $\theta_i = \beta_i^j + s \Delta$ for $s = \{-1, 1\}$. As demonstrated above, $\phi_i$ converges to $\beta_i^j$. Once $\phi_i = \beta_i^j$, $A_i$ will switch to the opposite correction, i.e., $\theta_i = \beta_i^j - s \Delta$. However, $\phi_j \neq \beta_j^i$ implies that $A_j$ will keep the previous direction. After that, $\dot{\beta}_1^2 \neq 0$ after $\phi_i \!=\! \beta_i$ such that the blocking is resolved. Thus, Corollary~\ref{th:unblock} is proved.

\subsection{Calculating blocking durations}~\label{ap:time}
\vspace{-2mm}
\begin{figure}[htp]
    \centering 
    \includegraphics[width=0.3\textwidth]{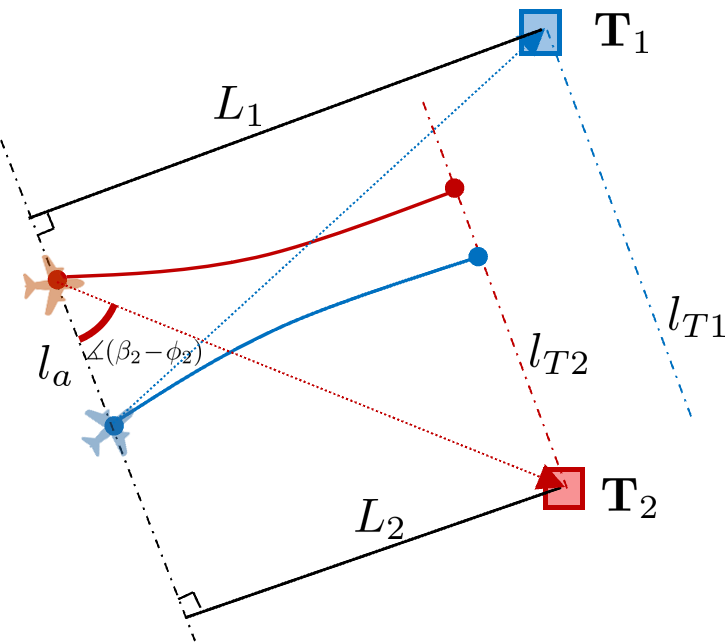}
    \caption{Geometric sketch of the blocking duration. Lines $l_a$, $l_{T1}$ and  $l_{T2}$ are parallel. }\vspace{-2mm}
    \label{fig:time_block}
\end{figure}

Colorary~\ref{th:unblock} implies that the current blocking mode will be resolved autonomously once the two airplanes become collinear with either target. As the geometric relationship shown in Fig.~\ref{fig:time_block}, the blocking duration relies on the earliest collinearity of $l_a$ with either  $L_{T_1}$ or $L_{T_2}$. Therefore, $T\!>\!\frac{\text{min}(L_1, L_2)}{v}$. As shown in the figure, $L_i\!=\!\|\mathbf{p}_i-\mathbf{T}_i\| \text{cos}(\beta_i^j-\phi_i)$ for $i\!\in\!\{1,2\}$. During the blocking, the distance of airplanes is converging to $r$. In terms of the triangle inequality theorem, the actual path is less than $\text{min}(L_1, L_2) + \|\mathbf{p}_1 - \mathbf{p}_2\| -r$. Therefore, $T < \frac{\text{min}(L_1, L_2) + \|\mathbf{p}_1 - \mathbf{p}_2\| -r}{v}$. To sum up, the lower and upper bound of the blocking duration are calculated to assess the impact of the blocking mode.

\subsection{Duration estimation for adaptive priority}~\label{ap:priority}
To determine the optimal unblocking priority, the durations of three options are evaluated for the ego airplane $A_i$. Based on the geometric relationship in Fig.~\ref{fig:time_cost}, the durations for the three options can be approximated as follows:
\begin{equation}\label{eq:time_est}
\begin{aligned}
T_{b}   \textstyle \approx \textstyle \frac{1}{v}[&2\min(\|\mathbf{p}_i\!-\!\mathbf{T}_i\| \text{cos}(\beta_i^j-\phi_i), \|\mathbf{p}_j\!-\!\mathbf{T}_j\| \text{cos}(\beta_j^i\!-\!\phi_j))\\
 & \textstyle + \sqrt{\|\mathbf{T}_i - \mathbf{o}\| - r^2} + \sqrt{\|\mathbf{T}_j - \mathbf{o}\|\!-\!r^2} + \pi r\,], \\ 
T_{u}^i   \textstyle \approx  \frac{1}{v}[ & \|\mathbf{T}_i - \mathbf{p}_j\| + \|\mathbf{T}_j - \mathbf{p}_j\|  + \pi r\,], \\ 
T_{u}^j   \textstyle \approx  \frac{1}{v}[ & \|\mathbf{T}_i - \mathbf{p}_i\| + \|\mathbf{T}_j - \mathbf{p}_i\|  + \pi r\,]
\end{aligned}
\end{equation}
where $\mathbf{o} = \frac{\mathbf{p}_1 + \mathbf{p}_2}{2}$. The adaptive priority scheme compares the approximate durations of three decisions and selects the optimal one for the blocking resolution strategy. Although the duration estimations are approximate and the final decision may not be optimal, the adaptive priority based on this duration estimation can statistically improve the utility.

\fi

\end{document}